\documentclass[12pt]{article}
\usepackage{amsfonts,amsthm,mathtools}

\theoremstyle{plain}
\theoremstyle{definition}
\newtheorem{theorem}{Theorem}[section]
\newtheorem{lemma}[theorem]{Lemma}
\newtheorem{definition-theorem}[theorem]{Definition-Theorem}
\newtheorem{definition-proposition}[theorem]{Definition-Proposition}

\newtheorem{proposition}[theorem]{Proposition}
\newtheorem{corollary}[theorem]{Corollary}
\newtheorem{example}{Example}[section]
\newtheorem{examples}{Example}[subsection]

\newtheorem{remark}{Remark}[section]
\newtheorem{remarks}{Remarks}[section]
\newtheorem{definition}{Definition}[section]

\numberwithin{equation}{section} 

\DeclareMathOperator{\sgn}{sgn}

\def\tr{\mathrm {tr}}
\def\det{\mathrm {det}}

\def\Pf{\mathrm {Pf}}

\def\res{\mathop{\mathrm {res}}\limits}

\def\res{\mathop{\mathrm{res}}\limits}

\def\&{&{\hskip -20pt}}

\def\be{\begin{equation}}
\def\ee{\end{equation}}

\def\bea{\begin{eqnarray}}
\def\eea{\end{eqnarray}}

\def\bt{\begin{theorem}}
\def\et{\end{theorem}}

\def\bex{\begin{example}\small \rm}
\def\eex{\end{example}}

\def\bexs{\begin{examples}\small \rm}
\def\eexs{\end{examples}}

\def\br{\begin{remark}\small \rm}
\def\er{\end{remark}}

\def\FF {{\mathcal F}}

\def\PP{{\mathcal P}}

\def\Zb{\mathbf{Z}}

\def\pb{\mathbf{p}}

\def\rb{\mathbf{r}}

\def\tb{\mathbf{t}}

\def\xb{\mathbf{x}}
\def\yb{\mathbf{y}}
\def\zb{\mathbf{z}}

\newcommand{\Pa}{\mathop\mathrm{P}\nolimits}
\newcommand{\DP}{\mathop\mathrm{DP}\nolimits}
\newcommand{\OP}{\mathop\mathrm{OP}\nolimits}

\def\bp{\begin{Proposition}\rm}
\def\ep{\end{Proposition}}
\def\bc{\begin{corollary}}
\def\ec{\end{corollary}}
\def\bl{\begin{lemma}\em}
\def\el{\end{lemma}}
\def\be{\begin{equation}}
\def\ee{\end{equation}}
\def\br{\begin{remark}\rm\small}
\def\er{\end{remark}}
\def\brs{\begin{remarks}.\\ \rm\
\begin{enumerate}}
\def\ers{\end{enumerate}\end{remarks}}
\def\bea{\begin{eqnarray}}
\def\eea{\end{eqnarray}}

\begin{document}

\begin{center}
\begin{Large}\fontfamily{cmss}
\fontsize{17pt}{27pt}
\selectfont
	\textbf{Notes about KP/BKP correspondence}
	\end{Large}
	
\bigskip \bigskip
\begin{large}
 A. Yu. Orlov\footnote[1]{e-mail:orlovs55@mail.ru}
 \end{large}
 \\
\bigskip

\begin{small}
{\em Shirshov Institute of Oceanology, Russian Academy of Science, 36 Nahimovskii Prospekt,
Moscow, Russia 117997, and Institute for Theoretical and Experimental Physics, Moscow, Russia  117218}
\end{small}
 \end{center}

To the memory of Sergei Mironovich Natanzon

\medskip

\begin{abstract}
I present a set of remarks related to joint works \cite{paper1},\cite{paper2},\cite{paper3},\cite{MMNO}. 
These are remarks about polynomials solutions and vertex operators, eigenproblem for polynomials and a remark 
related to the the conjecture of Alexandrov and Mironov, Morozov about the ratios of the projective Schur functions.
New results on the bilinear relations between characters of the symmetric and of the Sergeev group and on bilinear 
relations between skew Schur and projective Schur functions and also between shifted Schur and projective Schur
functions are added. Certain new matrix models are considered  and a comment on Mironov-Morozov-Natanzon cut-and-join 
relation is added.
\end{abstract}
\bigskip


\tableofcontents

\section{Introduction}

The goal of this notes is to present a number of remarks and observations concerning KP \cite{Sato},\cite{JM},\cite{MJD}
TL \cite{UT}, BKP (\cite{DJKM1}),\cite{KvdLbispec} and Veselov-Novikov \cite{VN}
(2DKP, see \cite{Taimanov}) tau functions. It is well-known that these hierarchies have a number of remarkable
applications in mathematics and physics.
The present paper adds some details to the works \cite{paper1},\cite{paper2},\cite{paper3} which in turn were based
on \cite{OS},\cite{Or},\cite{HLO},\cite{HLee},\cite{BHH}.
Sections \ref{review},\ref{small},  and Sections \ref{W}, \ref{evaluation} contains reviews
and add details to some pieces in \cite{paper1},\cite{paper2},
\cite{paper3} and some pieces in \cite{MMNO}.
Sections \ref{characters},\ref{skew-polynomials},\ref{shifted},\ref{vertex} contain new reslauts.

KP equation was integrated in \cite{Druma} and in the famous work of Zakharov and Shabat in 1974 \cite{ZakharovShabat}.
A great number of applications in mathematics and in mathematical physics
were done due to the works of Kyoto school \cite{Sato},\cite{JM},\cite{MJD}.
They introduce the concept of tau function and make use of the free fermions.
Using the free fermion approach they related the KP hierarchy to the $A_\infty$ root system and
expanded the integrability to different root systems. The analogue of the KP hierarchy related
to the root system $B_\infty$ is called BKP hierarhy. It was studied also in \cite{KvdLbispec}.

Let me recall the basic notions in very short.

But before that, I will make the following remark.
\br\label{notations}
 The notations in this article is slightly  
different from the notations in the works \cite{paper1},\cite{paper2},\cite{paper3}.
In the case of the KP, instead of the higher times $\tb=(t_1,t_2,t_3,\dots)$ common in solitons theory
the power sum variables are used: $\pb=(p_1,p_2,p_3,\dots)$, the relation is $mt_m=p_m$. 
Moreover, in the case of BKP the power sums $\pb^{\rm B}=(p_1,p_3,p_5,\dots)$ are related to the 
BKP higher times $\tb_B=(t_1,t_3,t_5,\dots)$ as follows: $mt_m=2p_m$. Instead of 
$\tau(\tb),s_\lambda(\tb),\gamma(\tb)$ and instead of $\tau^B(\tb_B),Q_\nu(\tfrac12\tb_B),\gamma^B(\tb_B)$
here we write respectively $\tau(\pb),s_\lambda(\pb),\gamma(\pb)$ and  
$\tau^{\rm B}(2\pb^{\rm B}),Q_\nu(\pb^{\rm B}),\gamma^{\rm B}(2\pb^{\rm B})$.
\er

\paragraph{Charged and neutral fermions.}
The fermionic creation and annihilation operators  satisfy the anticommutation  relations:
\be
[\psi_j,\psi_k]_+= [\psi^\dag_j,\psi^\dag_k]_+=0,\quad [\psi_j,\psi^\dag_k]_+=\delta_{jk} .
\label{charged-canonical}
\ee

We recall that a nonincreasing set of nonnegative integers $\lambda_1\ge\cdots \ge \lambda_{k}\ge 0$,
we call partition $\lambda=(\lambda_1,\dots,\lambda_{k})$, and $\lambda_i$ are called parts of $\lambda$.
The sum of parts is called the weight $|\lambda|$ of $\lambda$. The number of nonzero parts of $\lambda$
is called the length of $\lambda$, it will be denoted $\ell(\lambda)$. See \cite{Mac} for details.
Partitions will be denoted by Greek letters: $\lambda,\mu,\dots$. The set of all partitions is denoted by
$\Pa$. The set of all partitions with odd parts is denoted $\OP$.
Partitions with distinct parts are called strict partitions, we prefer
letters $\alpha,\beta$ to denote them. The set of all strict partitions will be denoted by $\DP$.
The Frobenius coordinated $\alpha,\beta$ for partitions $(\alpha|\beta)=\lambda\in\Pa$ are of usenames
(let me recall that the coordinates $\alpha=(\alpha_1,\dots,\alpha_k)\in\DP$ consists of the lengths of arms counted
from the main diagonal of the Young diagram of $\lambda$ while
$\beta=(\beta_1,\dots,\beta_k)\in\DP$ consists of the lengths of legs counted
from the main diagonal of the Young diagram of $\lambda$, $k$ is the length of the main diaginal of $\lambda$,
see \cite{Mac} for details).

The {\em vacuum} element $|n\rangle$ in each charge sector $\FF_n$ is the basis element
corresponding to the trivial partition $\lambda = \emptyset$:
\be
| n\rangle :=|\emptyset; n \rangle = e_{n-1} \wedge e_{n-2} \wedge \cdots .
\ee
Elements of the dual space $\FF^*$ are denoted as {\em bra} vectors $\langle w |$,
with the dual basis $\{\langle \lambda ;n|\}$ for $\FF^*_n$ defined by the pairing
\be
\langle \lambda; n | \mu; m\rangle = \delta_{\lambda \mu} \delta_{nm}.
\ee
For KP $\tau$-functions, we need only  consider the $n=0$ charge sector $\FF_0$,
and generally drop the charge $n$ symbol, denoting the basis elements simply as
\be
|\lambda\rangle :=|\lambda;0\rangle.
\ee
 For $j>0$, $\psi_{-j}$ and $\psi^\dag_{j-1}$
(resp. $\psi^\dag_{-j}$ and $\psi_{j-1}$) annihilate the right (resp. left) vacua:
\bea
\psi_{-j} |0\rangle &\&=0, \quad \psi^\dag_{j-1} |0\rangle =0, \quad \forall  j >0,
\label{vac_annihil_psi_j_r}
 \\
\langle 0| \psi^\dag_{-j}  &\&=0, \quad \langle 0 | \psi_{j-1}  =0, \quad \forall  j >0.
\label{vac_annihil_psi_j_l}
\eea

Neutral fermions $\phi^+_j$ and $\phi^-_j$ are defined \cite{DJKM1} by
 \be
\phi^+_j :=\frac{\psi_j+ (-1)^j\psi^\dag_{-j}}{\sqrt 2},\quad
\phi^-_j :=i\frac{\psi_j-(-1)^j \psi^\dag_{-j}}{\sqrt 2},\quad j \in \Zb
\label{charged-neutral}
\ee
(where $i=\sqrt{-1}$), and satisfy
\be
\label{neutral-canonical}
 [\phi^+_j,\phi^-_k]_+=0,\quad [\phi^+_j,\phi^+_k]_+ = [\phi^-_j,\phi^-_k]_+ =(-1)^j \delta_{j+k,0}.
\ee
In particular,
\be
(\phi^+_0)^2=(\phi^-_0)^2=\tfrac{1}{2}.
\label{square_phi_0}
\ee
Acting on the vacua $|0\rangle$ and $|1 \rangle$, we have
\bea
\phi^+_{-j}|0\rangle &\&= \phi^-_{-j} |0\rangle = \phi^+_{-j}|1\rangle = \phi^-_{-j} |1\rangle =0 , \quad  \forall  j > 0  , \quad \forall  j > 0,
\label{phi_vac_r} \\
\langle 0| \phi^+_{j} &\&= \langle 0|\phi^-_{j}  = \langle 1| \phi^+_{j} = \langle 1|\phi^-_{j}  =0 , \quad  \forall  j > 0 ,
\label{phi_vac1_l}
\\
\phi^+_0|0\rangle &\& =
- i \phi^-_0 |0\rangle =
\tfrac{1}{\sqrt{2}} \psi_0|0\rangle = \tfrac{1}{\sqrt{2}} |1\rangle  ,
\label{phi_0_ac_r}
\\
 \langle 0| \phi^+_0 &\& =
 i \langle 0|\phi^-_0  =
\tfrac{1}{\sqrt{2}} \langle  0 | \psi_0^\dag = \tfrac{1}{\sqrt{2}}\langle 1|.
\label{hatphi_0_vac_l}
\eea

Let us use notations
\bea\label{Psi-lambda}
\Psi_{\alpha,\beta} &:=&(-1)^{\sum_{j=1}^r\beta_j}(-1)^{\tfrac{1}{2}r(r-1)}
\psi_{\alpha_1}\cdots
\psi_{\alpha_r}\psi^\dag_{-\beta_1-1} \cdots \psi^\dag_{-\beta_r-1}\\
\label{Psi^dag-lambda}
\Psi^\dag_{\alpha,\beta} &:=&(-1)^{\sum_{j=1}^r\beta_j}(-1)^{\tfrac{1}{2}r(r-1)}
\psi_{-\beta_r-1}\cdots \psi_{-\beta_1-1} \psi^\dag_{\alpha_r} \cdots
\psi^\dag_{\alpha_1}
\eea
where $\lambda=(\alpha|\beta)$. Note that $|\lambda\rangle = \Psi_{\alpha,\beta}|0 \rangle$,
$\langle \lambda|=\langle 0|\Psi^\dag_{\alpha,\beta}$.

For $\alpha=(\alpha_1,\dots,\alpha_k)\in\DP$ introduce
\bea\label{Phi-alpha}
\Phi^\pm_\alpha &:=& 2^{\frac k2}\phi^\pm_{\alpha_1}\cdots \phi^\pm_{\alpha_{k}}\\  
\Phi^\pm_{-\alpha} &:=& (-1)^{\sum_{i=1}^k\alpha_i}2^{\frac k2}\phi^\pm_{-\alpha_k}\cdots \phi^\pm_{-\alpha_{1}}
\eea

We obtain
\be\label{<PsiPsi>}
\langle 0|\Psi^\dag_{\alpha,\beta}\Psi_{\alpha',\beta'} |0\rangle =\delta_{\alpha,\alpha'}\delta_{\beta,\beta'}
\ee
whose bosonized version is the scalar product of the Schur functions
$<s_{\alpha,\beta},s_{\alpha',\beta'}>=\delta_{\alpha,\alpha'}\delta_{\beta,\beta'}$,
see \cite{Mac},
and
\be\label{<PhiPhi>}
\langle 0|\Phi{^\pm}_{-\alpha}\Phi^{\pm}_{\alpha'}|0\rangle =2^{ \ell(\alpha)}\delta_{\alpha,\alpha'}
\ee
The bosonized version of (\ref{<PhiPhi>}) is the scalar product of the projective Schur functions:
$<Q_\alpha,Q_{\alpha'}> = 2^{ \ell(\alpha)}\delta_{\alpha,\alpha'}$

 In what follows sometimes we shall write $\phi_{x}$ and $\Phi_{x}$ instead of $\phi^+_{x}$ and $\Phi^+_{x}$.

\paragraph{Fermions and tau functions: KP and BKP cases.}
According to \cite{JM} KP tau functions can be presented in form of the following vacuum expectation value (VEV)
\be
\tau(\pb)=\langle 0|\gamma(\pb)g |0\rangle
\ee
where $g$ is an exponential of a bilinear in $\{\psi_i\}$ and $\{\psi^\dag\}$ expression.
Here
\be
\gamma(\pb)=e^{\sum_{m>0} \frac 1m p_mJ_m},\quad J_m=\sum_{i\in Z} \psi_i\psi^\dag_{i+m}
\ee
and $\pb$ is the set of parameters $(p_1,p_2,p_3,\dots)$, the numbers $\frac 1m p_m$ are called the KP higher times.

Similarly, the BKP tau function can be presented as
\be
\tau^{B}(2\pb^{\rm B})=\langle 0|\gamma^{\rm B\pm}(2\pb^{\rm B})h^\pm |0\rangle
\ee
where $h^\pm$ is an exponential of a quadratic in $\{\phi^\pm_i\}$  expression.
Here
\be
\gamma^{\rm B\pm}(2\pb^{\rm B})=e^{\sum_{m>0,{\rm odd}} \frac 2m p_m J^\pm _m},
\quad J^\pm_m=\sum_{i\in Z} (-)^i\phi^\pm_{-i-m}\phi^\pm_{i}
\ee
where $\pb^{\rm B}=(p_1,p_3,\dots)$.
 \br The set of $t_m^{\rm B}=2p_m^{\rm B},\, m=1,3,5,\dots $ is called the set of the BKP
 higher times.
 \er

After the the article \cite{Sato} symmetric functions appeared to be the part of the soliton theory,
these are power sum variables, Schur functions and later (see \cite{You} and \cite{Nim}) the projective
Schur functions.
 
The wonderful observation by Sato and his school \cite{Sato}, \cite{JM} is the fermionic formula for
the Schur polynomial
\be\label{Schur}
s_\lambda(\tb)=\langle 0|\gamma(\tb)\Psi_{\alpha,\beta}|0\rangle
\ee
In the BKP case the similar formula was found in \cite{You}:
\be\label{You}
Q_\alpha(\pb^{\rm B})=\langle 0|\gamma^{{\rm B}\pm}(2\pb^{\rm B}))\Phi^\pm_\alpha|0\rangle
\ee

\paragraph{Bosonization.} We recall that there exists the fermion-boson correspondence in 2D space. The first
work was the preprint of Pogrebkov and Sushko which preceded the article \cite{PogrebkovSushko}.
This correspondence turned out to be very important in the theory of solitons, as it was shown in a series
of excellent works of the Kyoto school. Let me recall some facts.

The fermionic Fock space is in the one-to-one correspondence with the bosonic Fock space which can be
viewed as the space of polynomials in a chosen set of parameters, say $\pb=(p_1,p_2,p_3,\dots)$ and a parameter
$\eta$:
\be
|\lambda; n\rangle \, \leftrightarrow \, s_\lambda(\pb) \eta^n
\ee
where $s_\lambda(\pb)$ is the Schur polynomials written in terms of the {\em power sum} variables $\pb$, see \cite{Mac}.
While the Fermi operators are in the one-to-one correspondence with the vertex operators:
\be
\psi(z)\,\leftrightarrow\,V^{+}(z),\quad \psi^\dag(z)\,\leftrightarrow\,V^{-}(z)
\ee
where
\be
V^\pm(z)=
e^{\pm\sum_{j=1}^\infty \frac 1jz^j p_j} e^{\mp \partial_n}z^{\pm n}
e^{\mp\sum_{j=1}^\infty  z^{-j} \frac{\partial}{\partial p_j}},\quad z\in S^1
\ee
where $e^{\mp \partial_n}z^{\pm n}$ is the shift operator which act on the variable $n$
( $n$ and $\partial_n$, $[\partial_n,n]=1$ are bosonic conjugated zero mode operators)
and
\be
\psi(z)=\sum_{i\in\mathbb{Z}} z^i\psi_i,\quad \psi^\dag(z)=\sum_{i\in\mathbb{Z}} z^{-i}\psi^\dag_i .
\ee

In the BKP case we have one-to-one correspondence between the Fock space of neutral fermions and the bosonic Fock
space which can be viewed as the space of polynomials in a set $\pb^{\rm B}=\frac12(t_1,t_3,t_5,\dots)$ and a 
Grassmannian parameter
$\xi$ ($\xi^2=0$)

\be
\Phi^\pm_\alpha|0 \rangle \, \leftrightarrow \, Q_\alpha(\pb^{\rm B}) \begin{cases}
                                                               1,\quad \ell(\alpha)\,{\rm even}\\
                                                               \xi,\quad \ell(\alpha)\,{\rm odd}
                                                              \end{cases}
\ee
where $Q_\lambda(\pb^{\rm B})$ is the projective Schur polynomials written in terms of the {\em power sum}
variables $\pb^{\rm B}$, see \cite{Mac}.
While the Fermi operators are in the one-to-one correspondence with the vertex operators:
\be
\phi^\pm(z)\,\leftrightarrow\,V^{\pm{\rm B}}(z)
\ee
where
\be
V^{{\rm B}\pm}(z)=\frac{1}{\sqrt{2}}(\xi +\frac{\partial}{\partial \xi} )
e^{\sum_{j=1,{\rm odd}}^\infty \frac 2j z^j p_j^\pm}
e^{-\sum_{j=1,{\rm odd}}^\infty  z^{-j} \frac{\partial}{\partial p^\pm_j}},\quad z\in S^1
\ee
and
\be\label{phi(z)}
\phi^\pm(z)=\sum_{i\in\mathbb{Z}} z^i\phi^\pm_i
\ee

Let us re-write formulas (\ref{Schur}) and (\ref{You}) for the both Schur functions:
\be
s_\lambda(\xb)\Delta(\xb)=\langle 0|\psi^\dag(x_1^{-1})\cdots \psi^\dag(x_N^{-1})\Psi_{\alpha,\beta} |N\rangle =
\langle N|\Psi^\dag_{\alpha,\beta} \psi(x_1)\cdots \psi(x_N) |0\rangle 
\ee
\be\label{Q(x)}
Q_\alpha(\xb)\Delta^*(\xb)=2^{-\frac N2}\langle 0|\phi(-x_1^{-1})\cdots \phi(-x_N^{-1})\Phi_\alpha |0\rangle =
 2^{-\frac N2}\langle 0|\Phi_{-\alpha} \phi(x_1)\cdots \phi(x_N) |0\rangle 
\ee
We use the following notations:
\be
\Delta(\xb):=\prod_{i<j}(x_i-x_j),\quad \Delta^*(\xb):=\prod_{i<j}\frac{x_i-x_j}{x_i+x_j}
\ee

\paragraph{Fermi fields}

For the Fermi fileds and for the Fourierr components of the Fermi fields we have quite similar relations:
\be\label{psi-phi-z}
\psi(-z^{-1}) =\frac{\phi^+(-z^{-1}) - i\phi^-(-z^{-1})}{\sqrt 2},\quad
\psi^\dag(z^{-1}) =\frac{\phi^+(-z^{-1})+ i \phi^-(-z^{-1})}{\sqrt 2},
\ee
\be
\psi_j =\frac{\phi^+_j - i\phi^- _{j}}{\sqrt 2},\quad
 (-1)^j \psi^\dag_{-j} =\frac{\phi^+_j+ i \phi^-_{j}}{\sqrt 2},
\ee
which result in
\be
\psi^\dag(z^{-1})\psi(-z^{-1})=-i\phi^+(-z^{-1})\phi^-(-z^{-1})
\ee
\be\label{psiphi-j}
(-1)^j\psi^\dag_{-j}\psi_j=-i\phi^+_j\phi^-_j
\ee

Next, we introduce
\bea\label{Psi-x}
\Psi(\xb,\yb):&=& \psi(x_1)\cdots \psi(x_N)\psi^\dag(-y_1) \cdots \psi^\dag(-y_N)
\\
\label{Psi-dag-x}
\Psi^*(\xb,\yb):&=&  \psi(-y_N^{-1})\cdots \psi(-y_1^{-1}) \psi^\dag(x_N^{-1}) \cdots \psi^\dag(x_1^{-1}) 
\eea
and
\bea\label{Phi-x}
\Phi^{*\pm}(\xb)= 2^{-\frac N2 }\phi^\pm(-x_1^{-1})\cdots \phi^\pm(-x_N^{-1})\\
\Phi^{\pm}(\xb)= 2^{-\frac N2 }\phi^\pm(x_1)\cdots \phi^\pm(x_N)
\eea
Here we consider $N$ to be an even number.

We will imply the ``time ordering'' in the products $\Psi,\Psi^\dag,\Phi^\pm$ above, namely, the condition
$|x_1^{-1}|>\cdots >|x_N^{-1}|$ and $|y_1^{-1}|>\cdots >|y_N^{-1}|$. Such ordering is similar to the usage of partitions 
as it was done in (\ref{Psi-lambda}) and (\ref{Psi^dag-lambda})

We have
\bea
\Psi(\xb,\xb)= (-1)^{\tfrac 12 N(N-1)}(-i)^N \Phi^+(\xb)\Phi^-(\xb)\\
\Psi^*(\xb,\xb)= (-1)^{\tfrac 12 N(N-1)}(-i)^N \Phi^{*+}(\xb)\Phi^{*-}(\xb)
\eea

Introduce
\be
p_m(\xb,\yb):=\sum_{i=1}^N \left( x^m - (-y)^m  \right)
\ee
where $N$ is even.
One can obtain
\be
s_\lambda(\xb,\yb)\Delta(\xb,\yb)=
\langle 0|\Psi^*(\xb,\yb)\Psi_{\alpha,\beta} |0\rangle \prod_{i=1}^N x_i^{-1}=
\langle 0|\Psi^\dag_{\alpha,\beta} \Psi(\xb,\yb) |0\rangle \prod_{i=1}^N y_i^{-1}
\ee
where
\be\label{Delta(x,y)}
\Delta(\xb,\yb)=(-1)^{\frac12 N(N-1)}\prod_{i=1}^N\frac{\Delta(\xb)\Delta(\yb)}{\prod_{i,j=1}^N(x_i+y_j)}
\ee
\be
Q_\alpha(\xb)\Delta^*(\xb)=\langle 0|\Phi(\xb)\Phi_\alpha |0\rangle =
\langle 0|\Phi_{-\alpha} \Phi^*(\xb) |0\rangle 
\ee
Relation (\ref{Delta(x,y)}) for $\Delta(\xb,\yb)$ should be compared with the relation
\be
s_{(\alpha|\beta)}(\pb_1)=\frac{\Delta(\alpha)\Delta(\beta)}{\prod_{i,j=1}^N(\alpha_i+\beta_j+1)}
\prod_{i=1}^r \frac{1}{\alpha_i!\beta_i!}=\frac{\dim\,\lambda}{|\lambda|!}
\ee
where $\lambda=(\alpha|\beta)$, ($\dim\,\lambda$ the number of ways to build the $ \lambda $ Young diagram by adding 
one node to a previous diagram, starting with the empty diagram, so that at each step we will have a Young diagram) 
and where
\be
\pb_1=(1,0,0,\dots).
\ee
\br
In literature (for instance in \cite{Mac}) one meets the notation $s_\lambda(x/y)$ which $s_\lambda(\xb,-\yb)$
in our present notations. Our notations are convenient for our purposes.
\er

\section{On KP vs BKP correspondence in \cite{paper1},\cite{paper2},\cite{paper3} \label{review}}

\subsection{Preliminaries.}
Let the KP higher times are parametrized by
 \be
 t_m=t_m([\xb,\yb]) = \sum_{i=1}^N \left( x_i^m - (-y_i)^m   \right)
 \ee
 and let the Frobenius parts of the partition $\lambda = (\alpha|\beta)$ be denoted by
 $\alpha=(\alpha_1,\dots,\alpha_r)\in\DP$ and $\beta=(\beta_1,\dots,\beta_r)\in \DP$.
 Our formulas are of the following type:
 \be\label{s-QQ}
 s_{(\alpha|\beta)}([\xb,\yb])=\sum_{P_x} \sum_{P_F} A(P_F) D(P_x) Q_{\nu^+}({\bf z}^+) Q_{\nu^-}({\bf z}^-)
 \ee
 and more generally \cite{paper4}:
 \be\label{KP-BKPBKP}
 \tau^{KP}_{(\alpha|\beta)}([\xb,\yb])=\sum_{P_x} \sum_{P_F} A(P_F) D(P_x) \tau^{BKP}_{\nu^+}({\bf z}^+)
 \tau^{BKP}_{\nu^-}({\bf z}^-)
 \ee
 where the ${\bf z}^+$ is a subset of the set ${\bf z}$ of coordinates $x_1,\dots,x_M,y_1,\dots, y_M$ and ${\bf z}^-$
 is the complementary subset ${\bf z}\setminus {\bf z}^+$ and where $A(P_F)$ and  $D(P_x) $ is a special combinatorial
 numbers, see\cite{paper1} and  $D(P_x) $ is
 basically the same number where the partitions are
   $ \alpha, \beta, \nu^+, \nu^- $ are replaced by continuous variables
    $ \xb, \yb, {\bf z}^+, {\bf z}^- $, and also contain a Vandermonde-type factor that depends on these 
    continuous variables.
 The partition $\nu^+$ is obtained
 as an ordered subset of the set $F$ numbers $\alpha_1,\dots,\alpha_r,\beta_1+1,\dots,\beta_r+1$
 and the partition $\nu^-$ is the complementary subset $F\setminus \nu^+$ under the condition:
 in case a pair of equal numbers occur in the set $F$ (say, $\alpha_i=\beta_j+1$), then these numbers
 belong to {\it different} subsets (these are either $\nu^+$ or $\nu^-$). Symbols $P_F$ and $P_x$
 which we call polarizations denote the selection of the subsets.

 The papers \cite{paper1},\cite{paper2} deal with finding the weights $a(P_F)$ respectively in cases (\ref{s-QQ})
 and \ref{KP-BKPBKP}, the weights coincide. The work \cite{paper4} deals with the determing of $d(P_x)$ in both
 cases.

 We also deal with formulas which relate the polynomials naturally obtained in the KP theory
 and polynomials naturally obtained the BKP theory.

\subsection{Four key lemmas.}

There are four statements that follows from (\ref{psi-phi-z})-(\ref{psiphi-j}) and from
(\ref{vac_annihil_psi_j_r})-(\ref{vac_annihil_psi_j_l})
and (\ref{phi_vac_r})-(\ref{hatphi_0_vac_l} ). We need certain notations.

For a given partition $\lambda=(\alpha|\beta)$,
we have an ordered set of numbers $(\alpha,I(\beta)) =\alpha_1,\dots,\alpha_r,\beta_1+1,\dots,\beta_r+1$, which
consists of two
strict partitions of equal length:
left $\alpha=(\alpha_1,\dots,\alpha_r)$ and right $I(\beta)=(\beta_1+1,\dots,\beta_r+1)$.
Let us perform a permutation $\sigma$ in this set, which satisfies the following
conditions:

\begin{itemize}
 \item[1]

 the resulting sequence of numbers again consists of two consecutive strict partitions $(\nu^+,\nu^-)$:
a left strict partition $\nu^+=(\nu^+_1,\dots,\nu^+_{p_1})$, $\nu^+_1>\dots >\nu^+_{p_1}\ge 0$ and a right strict
partition $\nu^-=(\nu^-_1,\dots,\nu^-_p)$, $\nu^-_1>\dots >\nu^-_{p_2}\ge 0$. We call $p_1$ (or $p_2$ ) the cardinality
of partition $\nu^+$ (of $\nu^-$) and denote it by $m(\nu^+)$ (by $m(\nu^-)$)

\item[2]
in the original set $(\alpha,\beta)$  there can be pairs of
equal numbers (for example $\alpha_i=\beta_j+1$). Let us denote by $s$ the number of such pairs. It is required
that the element of each pair that was in the left partition (that is, in
$\alpha$) remains in the left partition $\nu^+$
(accordingly, the right element of the pair ends up in the right partition $\nu^-$).
\end{itemize}

More notations:
\begin{itemize}
 \item We call the pairs $(\nu^+,\nu^-)$ {\em polarization} of $(\alpha|\beta)$.  For a given $\lambda=(\alpha|\beta)$,
we denote $P(\alpha,\beta)$ set of all possible pairs $(\nu^+,\nu^-)$ under the conditions above
  \item We denote $\sgn \sigma (\nu^+,\nu^-)$ the sign of the permutation from the ordered set $(\alpha,I(\beta))$
  to the ordered set $(\nu^+,\nu^-)$
  \item We denote $\pi(\nu^\pm) :=\#(\alpha \cap \nu^{\pm})$
the cardinality of the intersection of $\alpha$ with $\nu^\pm$
\end{itemize}
(See also Appendix for more formal treatment.)

Introduce
\be\label{textsc{a}}
 \textsc{a}^{\nu^+,\nu^-}_{\alpha,\beta}:={(-1)^{\tfrac{1}{2}r(r+1) + s}\over 2^{r-s}}
\sgn(\nu)(-1)^{\pi(\nu^-)} i^{ m(\nu^-)}
\ee

\begin{lemma}\label{Polarization-partitions}
\bea\label{Psi=P-P+}
\Psi_{(\alpha|\beta)}&=&
\sum_{(\nu^+,\nu^-)\in \PP(\alpha, \beta)} \textsc{a}^{\nu^+,\nu^-}_{\alpha,\beta} \Phi^+_{\nu^+}\Phi^-_{\nu^-} ,
\\
\label{Psi=P-P+dag}
\Psi^\dag_{(\alpha|\beta)}&=&\sum_{(\nu^+,\nu^-)\in \PP(\alpha, \beta)} (-1)^{|\nu^+|+|\nu^-|}
\textsc{a}^{\nu^+,\nu^-}_{\alpha,\beta} \Phi^+_{-\nu^+}\Phi^-_{-\nu^-}
\eea

\end{lemma}

  The proof of (\ref{Psi=P-P+})
 takes into account the sign factor $\sgn(\nu)$ corresponding to the order of the neutral fermion
factors, as well as the powers of $-1$ and $i$, and noting that there are $2^s$  resulting identical terms,
then obtain our result, see also (\ref{polariz_sum}). The sign factor in (\ref{Psi=P-P+dag}) can be verified via 
$\langle 0| \Psi^\dag_{(\alpha|\beta)}\Psi_{(\alpha'|\beta')}|0\rangle =\delta_{\alpha,\alpha'}\delta_{\beta,\beta'}$,
see (\ref{<PsiPsi>})-(\ref{<PhiPhi>}).

Thanks to the quite similar relations (\ref{psi-phi-z})-(\ref{psiphi-j}) we obtain the direct
analogue of Lemma \ref{Polarization-partitions}:
\bl\label{Polarization-fields}
\be
\Psi(\xb,\yb)=
\sum_{(\zb^+,\zb^-)\in \PP(\xb, \yb)} \textsc{a}^{\zb^+,\zb^-}_{\xb,\yb} \Phi^+(\zb^+)\Phi^-(\zb^-) ,
\ee

\el
This is used in the forthcoming work \cite{paper4}.

Let us introduce
\be\label{J-Delta}
J_\Delta := \prod_{i=1}^{\ell(\Delta)} J_{\Delta_i},\qquad
J_{-\Delta} := \prod_{i=1}^{\ell(\Delta)} J_{-\Delta_i},\quad \Delta\in\Pa
\ee
and
\be\label{J-Delta-B}
J^{{\rm B}\pm}_\Delta := \prod_{i=1}^{\ell(\Delta)} J^{{\rm B}\pm}_{\Delta_i},\qquad
J^{{\rm B}\pm}_{-\Delta} := \prod_{i=1}^{\ell(\Delta)} J^{{\rm B}\pm}_{-\Delta_i},\quad \Delta\in\OP
\ee

For a given $\Delta\in\OP$, one can split its parts into two ordered odd  partitions $(\Delta^+,\Delta^-)$:
$\Delta=\Delta^+\cup\Delta^-$, $\Delta^+,\Delta^- \in OP$,
$\ell(\Delta^+) + \ell(\Delta^-)=\ell(\Delta)$. The set of all such  $(\Delta^+,\Delta^-)$ we denote
$O\PP(\Delta)$.

From $J_n=J^{\rm B+}_n+J^{\rm B^-}_n,\,n\,{\rm odd}$ (see \cite{JM}), we obtain
\bl \label{Polarization-currents}
\be
J_\Delta = \sum_{(\Delta^+,\Delta^-)\in O\PP} J^{{\rm B}+}_{\Delta^+}J^{{\rm B}-}_{\Delta^-},
\qquad
J_{-\Delta} = \sum_{\Delta^+\in \OP\atop \Delta^+\cup \Delta^+=\Delta} J^{{\rm B}+}_{-\Delta^+}J^{{\rm B}-}_{-\Delta^-},
\ee
\el

 \begin{lemma}[\bf Factorization]
 \label{factorization_lemma}
 If $U^+$ and $U^-$ are either even or odd degree elements of the subalgebra generated by the operators
 $\{\phi^+_i\}_{i \in \Zb}$ and  $\{\phi^-_i\}_{i \in \Zb}$ respectively, the VEV of their product can be factorized as:
\be
\langle 0 | U^+ U^-|0\rangle =
\begin{cases}
\langle 0 |  U^+|0\rangle \langle 0| U^- |0\rangle
&
\text{ if } U^+ \text{ and } U^-\text{ are both of even degree}\\
0  &  \text{ if } U^+ \text{ and }  U^- \text{ have different  parity }\\
2i \langle 0|U^+\phi^+_0|0\rangle \langle 0|U^-\phi^-_0 |0\rangle &
\text{  if }  U^+ \text{ and }  U^+ \text{ are both of odd degree}.
\end{cases}
\label{gen_factorization}
\ee
\end{lemma}

In \cite{paper1},\cite{paper2},\cite{paper3} we apply Lemma \ref{Polarization-partitions} and
Lemma \ref{factorization_lemma} to evaluate vacuum expectation values. For these purposes it
was suitable to have

\begin{definition}{\bf Supplemented partitions.}
If $\nu$ is a strict partition of cardinality $m(\nu)$ (with $0$ allowed as a part),
we define the associated {\em supplemented partition} $\hat{\nu}$ to be
\be
\hat{\nu} := \begin{cases}  \nu \ \text{ if } m(\nu) \ \text{ is even}, \cr
                  (\nu,0)  \  \text{ if } m(\nu) \  \text{ is odd}.
                  \end{cases}
 \label{hat_nu}
\ee
We denote by $m(\hat{\nu})$  the cardinality of $\hat{\nu}$.
\end{definition}

For instance, as a result of application of Lemmas \ref{Polarization-partitions},\ref{factorization_lemma}
to (\ref{Schur}) and (\ref{You}) we get \cite{paper1}
\be\label{sQQ}
s_{(\alpha|\beta)}(\tb')=\sum_{(\nu^+,\nu^-)\in \PP(\alpha,\beta)}
{a}^{\nu^+,\nu^-}_{\alpha,\beta} Q_{\nu^+}(\pb^{\rm B})Q_{\nu^-}(\pb^{\rm B})
\ee
where the power sum variables are as follows: $\tb'=(t_1,0,t_3,0,t_5,0,\dots)$,
$\pb^{\rm B}=\frac12(t_1,t_3,t_5,\dots)$,
and where
\be\label{a}
 {a}^{\nu^+,\nu^-}_{\alpha,\beta}:={(-1)^{\tfrac{1}{2}r(r+1) + s}\over 2^{r-s}}
\sgn(\nu)(-1)^{\pi(\nu^-)+\frac12  m(\hat{\nu}^-)}
\ee
which is obtained from (\ref{textsc{a}}) by replacing the factor $ i^{  m({\nu}^-)} $
by $ (-1)^{\frac12  m(\hat{\nu}^-)}$
(notice the hat above $\nu^-$
- this is the result of the application of Lemma \ref{factorization_lemma}).

\section{Relation between characters of symmetric group and characters of Sergeev group \label{characters}}

It is known \cite {Mac} that the power sums labeled by partitions:
\be
\pb_\Delta =p_{\Delta_1}p_{\Delta_2}\cdots,\quad \Delta\in\Pa
\ee
(here $ \pb = (p_1, p_2, p_3, \dots) $) are uniquely expressed
in terms of Schur polynomials
\be\label{p-s-chi}
\pb_\Delta =\sum_{\lambda\in\Pa} \chi_\lambda(\Delta) s_\lambda(\pb),\quad
\ee
while the odd power sum variables (power sums labels by odd numbers)
\be
\pb_\Delta =p_{\Delta_1}p_{\Delta_2}\cdots,\quad \Delta\in\OP
\ee
also denoted $\pb^{\rm B}_\Delta$ (where $\pb^{\rm B}=(p_1,p_3,p_5,\dots)$)
 are uniquely expressed in terms of projective Schur polynomials:
\be\label{pbB-Q}
\pb^{\rm B}_\Delta =\sum_{\alpha\in\DP} \chi^{\rm B}_\alpha(\Delta) Q_\alpha(\pb^{\rm B})=\pb_\Delta,
\quad \Delta \in \OP
\ee

Let me recall that the coefficients $\chi_\lambda(\Delta)$ in (\ref{p-s-chi}) has the meaning of
the irreducible characters of the symmetric group  $S_d,\,d=|\lambda|$ evaluated on the cycle class $C_\Delta$
$\Delta=(\Delta_1,\dots,\Delta_k) $, $|\lambda|=|\Delta|=d$, see \cite{Mac}, and we can write it as
\be
\chi_\lambda(\Delta) = \langle 0| J_\Delta   \Psi_{\alpha,\beta} |0\rangle
\ee
where $J_\Delta $ and $\Psi_{\alpha,\beta}$ are given respectively by (\ref{J-Delta}) and (\ref{Psi-lambda})
(see for instance \cite{MMNO} for details).

The characters of symmetric groups they have a very wide application, 
in particular, in mathematical physics. I will give two works as an example \cite{ItMirMor},\cite{MMS-knots}.

The notion of Sergeev group was introduced in \cite{EOP}, see Appendix. 
The coefficient $\chi^{\rm B}_\alpha$ is the irreducible character of this group \cite{EOP},\cite{Serg}.
As it was shown in \cite{EOP} the so-called spin Hurwitz numbers (introduced in this work) are expressed in terms
of these characters.
As it was pointed out in \cite{Lee2018},\cite{MMN2019} the generating function for spin Hurwitz numbers
can be related to the BKP hierarchy in a similar way as the generating function for usual Hurwitz numbers
is related to the KP (and also to Toda lattice) hierarchies, see pioneering works
\cite{Okoun2000},\cite{Okounkov-Pand-2006},
\cite{Goulden-Jackson-2008} and also
\cite{MMN},\cite{MM4},\cite{AMMN-2014},\cite{Uspehi-KazarianLando},\cite{HO-2014},\cite{CheAmb},\cite{NO-LMP},
\cite{Harnad-overview-2015}.

We can write these characters in terms of the BKP currents $J^{\rm B}_m$ ($m$ odd):
\be
   \chi^{\rm B}_\alpha(\Delta) =2^{-\ell(\alpha)} \langle 0|J^{\rm B}_{\Delta}\Phi_{\alpha} |0\rangle
\ee
see \cite{MMNO}.

From Lemma \ref{Polarization-partitions} and (\ref{Polarization-currents}) we obtain
\begin{theorem}
The character $\chi_\lambda,\,\lambda=(\alpha|\beta)$ evaluated on an odd cycle $\Delta\in\OP$  is the bilinear
function of the Sergeev characters as follows:
\be
\chi_{(\alpha|\beta)}(\Delta)=
\sum_{  (\nu^+,\nu^-)\in\PP(\alpha,\beta) \atop (\Delta^+,\Delta^-)\in O\PP(\Delta)} 2^{\ell(\nu^+) + \ell(\nu^+) }
a^{\nu^+,\nu^-}_{\alpha,\beta}
\chi^{\rm B+}_{\nu^+}(\Delta^+)\chi^{\rm B-}_{\nu^-}(\Delta^-)
\ee
where $a^{\nu^+,\nu^-}_{\alpha,\beta}$  are given by (\ref{a}).
\end{theorem}

\section{Relation between generalized skew Schur polynomials and generalized skew projective Schur polynomials
\label{skew-polynomials}}

This section may be treated as a remark to \cite{paper1}.

Let us find the relation between the following quantities:
\bea\label{skew-s}
s_{\lambda/\mu}(\pb')&:=&\langle 0|\Psi^\dag_\mu\gamma(\pb') \Psi_{\alpha,\beta} |0\rangle,\\
\label{skew-Q}
Q_{\nu/\theta}(\pb^{\rm B})&:=&\langle 0|\Phi_{-\theta}\gamma^{\rm B}(2\pb^{\rm B}) \Phi_\nu|0\rangle
\eea
where the power sum variables are as follows: $\pb'=(p_1,0,p_3,0,p_5,0,\dots)$, $\pb^{\rm B}=\frac12(p_1,p_3,p_5,\dots)$,
where $\lambda=(\alpha|\beta)$, $\mu=(\gamma|\delta)$ and where $\alpha,\beta,\gamma,\delta,\theta,\nu\in\DP$.

 \begin{theorem}\label{skew}
 \be
 s_{\lambda/\mu}(\pb') = \sum_{(\nu^+,\nu^-)\in \PP(\alpha,\beta)\atop (\theta^+,\theta^-)\in \PP(\gamma,\delta)}
 (-1)^{|\theta^+|+|\theta^-|}
 a^{\nu^+,\nu^-}_{\alpha,\beta} a^{\theta^+,\theta^-}_{\gamma,\delta}
 Q_{\nu^+/\theta^+}(\pb^{\rm B})Q_{\nu^-/\theta^-}(\pb^{\rm B})
 \ee
 \end{theorem}

 For the proof we apply Lemma \ref{Polarization-partitions} to the first equality (\ref{skew-s})
 and consider all possible parities of $m(\nu^\pm),m(\theta^\pm)$ to apply Lemma \ref{factorization_lemma}
 and take into account the fermionic expression (\ref{skew-Q}).

 Theorem \ref{skew} in a straightforward way can be generalized for the generalized skew Schur and
 skew projective Schur polynomials which we define as follows:
 Suppose 
 \be
 g =g(C)=e^{\sum C_{ij}\psi_i \psi^\dag_j }
 \ee
 where the entries $C_{ij}$ form a matrix.
  Also suppose
 \be
 h^\pm=h^\pm(A)=e^{\sum A_{ij}\phi_i^\pm \phi_j^\pm}
 \ee
 where now the entries $A_{ij}=-A_{ji}$ form a skew-symmetric matrix.
Suppose that
\be\label{Borel}
g |0\rangle = |0\rangle c_1,
\quad
h^\pm |0\rangle = |0\rangle c_2 ,\quad c_{1,2}\in \mathbb{C},\quad c_{1,2}\neq 0.
\ee

 Define the generalized Schur polynomials and the generalized projective Schur polynomials by
\bea\label{skew-s-A}
s_{\lambda/\mu}(\pb|g)&:=&\langle \mu|\gamma(\pb)g |\lambda\rangle,\\
\label{skew-Q-A}
Q_{\nu/\theta}(\pb^{\rm B}|h^\pm)&:=&\langle 0|\Phi_{-\theta}\gamma^{\rm B\pm}(2\pb^{\rm B}) h^\pm \Phi_\nu|0\rangle
\eea
\br
In (\cite{paper2}) $s_{\lambda/\mu}(\pb|g)$ was defined $s_{\lambda/\mu}(\pb|C)$ and 
$Q_{\nu/\theta}(\pb^{\rm B}|h^\pm)$ was defined $Q_{\nu/\theta}(\pb^{\rm B}|A)$.
\er

The polynomiality of (\ref{skew-s-A})  in $p_1,\dots,p_{|\lambda|-|\mu|}$ and the polynomiality
of  (\ref{skew-Q-A}) in $p_1,\dots,p_{|\nu|-|\theta|}$
follows from (\ref{Borel}).

One can treat a given $\pb$ as a constant and study $s_{\lambda/\mu}(\pb|A)$ as the function of discrete
variables $\lambda_i-i$ and $\mu_i-i$.

 \begin{theorem}\label{skew-A} Suppose $g=h^+h^-$ and (\ref{Borel}) is true. Then
 \be
 s_{\lambda/\mu}(\pb'|g) = \sum_{(\nu^+,\nu^-)\in \PP(\alpha,\beta)\atop (\theta^+,\theta^-)\in \PP(\gamma,\delta)}
 (-1)^{|\theta^+|+|\theta^-|}
 a^{\nu^+,\nu^-}_{\alpha,\beta} a^{\theta^+,\theta^-}_{\gamma,\delta}
 Q_{\nu^+/\theta^+}(\pb^{\rm B}|h^+)Q_{\nu^-/\theta^-}(\pb^{\rm B}|h^-)
 \ee
 \end{theorem}
The proof is based on the same reasoning as the proof of the theorem  in \cite{paper2}. We omit it.

 Introduce
\bea\label{factorial-s}
s_{\lambda/\mu}(\pb|\rb)&:=&\langle 0|\Psi^\dag_\mu \gamma_r(\pb) \Psi_{\alpha,\beta}|0\rangle,\\
\label{factorial-Q}
Q_{\nu/\theta}(\pb^{\rm B}|\rb)&:=&\langle 0|\Phi_{-\theta}\gamma_r^{\rm B\pm}(2\pb^{\rm B}) \Phi_\nu|0\rangle
\eea
where
\be\label{gamma_r-KP}
\gamma_r(\pb):=e^{\sum_{j=1}^\infty \frac 1j p_j \sum_{k\in\Zb} \psi_k \psi^\dag_{k+j} r(k+1)\cdots r(k+j)}
\ee
\be\label{gamma_r-BKP}
\gamma_r^{\rm B}(2\pb^{\rm B}):=e^{\sum_{j=1,{\rm odd}}^\infty \frac 2j p_j
\sum_{k\in\Zb} (-1)^k\phi^\pm_{k-j} \phi^\pm_{-k} r(k+1)\cdots r(k+j)}
\ee
By the Wick theorem we obtain
\be
 s_{\lambda/\mu}(\pb|\rb)=\det \left( r(\mu_i-i+1) \cdots r(\lambda_j-j)
 s_{(\lambda_i-\mu_j-i+j)}(\pb)  \right)_{i,j}
\ee
The similar relation for $Q_{\nu/\theta}(\pb^{\rm B}|\rb)$ is more spacious, the Wick theorem
yields a Nimmo-type pfaffian formula, we
shall omit it. We have
\be
 s_{\lambda/\mu}(\pb'|\rb) = \sum_{(\nu^+,\nu^-)\in \PP(\alpha,\beta)\atop (\theta^+,\theta^-)\in \PP(\gamma,\delta)}
 (-1)^{|\theta^+|+|\theta^-|}
 a^{\nu^+,\nu^-}_{\alpha,\beta} a^{\theta^+,\theta^-}_{\gamma,\delta}
 Q_{\nu^+/\theta^+}(\pb^{\rm B}|\rb)Q_{\nu^-/\theta^-}(\pb^{\rm B}|\rb)
 \ee

\subsection{Relation between shifted Schur and shifted projective Schur functions \label{shifted}}

Let us recall the notion of the shifted Schur function introduced by Okounkov and Olshanski \cite{OkOl}.
In can be defined as
\be\label{s^*}
s^*_\mu(\lambda) = \frac{\dim\, \lambda/\mu}{\dim\,\lambda} n(n-1)\cdots (n-k+1)
=\frac{s_{\lambda/\mu}(\pb_1)}{s_{\lambda}(\pb_1)}
\ee
where $n=|\lambda|$, $k=|\mu|$, $\pb_1=(1,0,0,\dots)$ and
$$
\dim\, \lambda/\mu=s_{\lambda/\mu}(\pb_1)(n-k)!,\quad \dim\, \lambda=s_{\lambda}(\pb_1)n!.
$$
are the number of the standard tableaus of the shape $\lambda/\mu$ and $\lambda$ respectively, see \cite{Mac}.
The function $s^*_\mu(\lambda)$ as a function of the Frobenius coordinates is also known as Frobenius-Schur
function $FS(\alpha,\beta)$

On the other hand, Ivanov \cite{Iv} introduced the projective analogue of shift $Q$-functions:
\be\label{Q^*}
Q^*_\theta(\nu)
=\frac{Q_{\nu/\theta}(\pb_1)}{Q_{\nu}(\pb_1)}
\ee
Therefore,
\be
s^*_\mu(\lambda)s_{\lambda}(\pb_1)=
\sum_{(\nu^+,\nu^-)\in \PP(\alpha,\beta)\atop (\theta^+,\theta^-)\in \PP(\gamma,\delta)}
 (-1)^{|\theta^+|+|\theta^-|}
 a^{\nu^+,\nu^-}_{\alpha,\beta} a^{\theta^+,\theta^-}_{\gamma,\delta}
 Q^*_{\theta^+}(\nu^+)Q^*_{\theta^-}(\nu^-)
 Q_{\nu^+}(\pb_1)Q_{\nu^-}(\pb_1)
\ee

Let me add that both $s^*$ and $Q^*$ were used in the description of the generalized cut-and-join structure
\cite{MMN2011},\cite{MMN2019} 
 in the topics of Hurwitz and spin Hurwitz numbers.

\section{Small remarks to \cite{paper2},\cite{paper3} \label{small}}

\subsection{A note to the articles \cite{paper2},\cite{paper3} about polynomial solutions}

A simple,
in fact, a trivial remark that should be made nonetheless.
When we say "polynomial solutions", then we should always clarify what variables we are talking about.

Consider any KP tau function written in Sato form
\be\label{Sato-form} 
\tau(\pb)=\sum_{\lambda\in\Pa} \pi_\lambda s_\lambda(\pb)
\ee
where $\Pa$ is the set of all partitions,
 $s_\lambda$ are Schur functions \cite{Mac} and $\pi_\lambda$ are Plucker coordinates of the point
in the Sato Grassmannian \cite{Sato}. For the tau function (\ref{Sato-form}) to be a polynomial in the 
variables $p_1,p_2,\dots$, it is obviously necessary that only a finite number of $\pi_\lambda$ be nonzero. 
But more often we mean symmetric polynomials in $x_1,x_2,\dots$ variables that arise from the substitutions
\be
p_m = \pm \sum_{i=1}^N x_i^m
\ee

Suppose that $\pi_\lambda =0$ if the length $\ell(\lambda)>M$.
Then it is enough to choose
$$
p_m = p_m(\xb) = - \sum_{a=1}^N x_a^m
$$
then $\tau(\pb(x))$ is a symmetric polynomial in the variables $\xb=(x_1,\dots,x_N)$ of the weight
at most $MN$.
The example is the hypergeometric family of tau function.
Such polynomils can have a determinant representation in form
\be\label{det-i-x_j}
\tau(\pb(\xb)) = \langle 0|\gamma(\pb(\xb))\,g\,|0\rangle
=\frac{1}{\Delta(\xb)}\det \left( \langle 0|\psi^\dag_i\psi(x_j^{-1}) g|0\rangle \right)|_{i,j}
\ee
which can be derived from the bosonization formulas and from the Wick theorem.
However, whether the last equality is true  or not true depends on the choice of $g$.
In the next Section \ref{fermionic_wick}, we just discuss for what $ g $ Wick's theorem can be applied, in form written
in (\ref{det-i-x_j}).

Similarly, any BKP tau function can be written in form
\be
\tau^{\rm B}(2\pb^{\rm B})=\sum_{\alpha\in\DP} A_\alpha Q_\alpha(\pb^{\rm B})
\ee
where $Q_\alpha$ are the projective Schur functions \cite{Mac} and $A_\alpha$ are Cartan coordinates of the point
in the isotropic Grassmannian \cite{HB}, $\DP$ denote the set of all strict partitions (partitions with distinct parts).

Suppose that $A_\lambda =0$ if the length $\ell(\lambda)>M$.
Then it is enough to choose
\be
p_m = p_m(\xb) = - 2\sum_{a=1}^N x_a^m
\ee

As examples I can mention orthogonal polynomials in Appendix 6 of \cite{OS2000}, 
these are: $q$-Askey-Wilson polynomials, coninuous $q$-Jacoby polynomials, 
$q$-Gegenbauer polynomials,  Clebsh-Gordan coefficients $C_q$,
$q$-Hahn polynomials. $q$-Racach polynomials.

\subsection{How to use the Wick theorem, a remark to \cite{paper3}
\label{fermionic_wick}}

Another simple remark can be added about the usage of the Wick theorem.
The Wick theorem is the main tool to get
various pfaffian and determinant expressions for multivariable polynomials which naturally appear
in the framework of the KP and the BKP hierarchies.

For an even number of fermionic operators $(w_1, \dots, w_{2L})$ that anticommute:
\be
[w_j, w_k]_+=0, \quad 1 \le j, k \le 2L
\label{anticommute}
\ee
the matrix with elements  $\langle 0|w_jw_k |0\rangle$ is skew symmetric, and Wick's theorem  implies that
$\langle 0|w_1\cdots w_{2L} |0\rangle$ is given by its Pfaffian
\be
\langle 0|w_1\cdots w_{2L} |0\rangle = \Pf\left( \langle 0|w_jw_k |0\rangle \right).
\label{Wick-Pf}
\ee
On the other hand,  if  the odd elements  $w_1,w_3,\dots$ are linear combinations of creation
operators $\{\psi_j\}_{j\in\Zb}$ and the even ones
$w_2, w_4,\dots$ linear combinations of annihilation operator $\{\psi_j^\dag\}_{j\in\Zb}$,  Wick's theorem implies
\be
\langle 0|w_1\cdots w_{2L} |0\rangle = \det\left( \langle 0|w_jw_k |0\rangle \right)_{j=1,3,\dots; k=2,4,\dots}.
\label{Wick-det}
\ee

Then, the problem is to present the Fock vector $\gamma(\pb)g|0\rangle$ in the factorized form.

It is possible in case we can factorize $g$ as follows:
\be\label{fact}
g=g^+ g^0 g^-,
\ee
where
\be
g^-|0\rangle =|0\rangle,\quad g^0|0\rangle =c|0\rangle,\quad \langle 0|g^+ = \langle 0|
\ee

However it may be impossible. Say if  $g=O_\lambda$ where $O_\lambda$ is the operator which creates the
state
\be
O_\lambda |0\rangle = |\lambda\rangle
\ee
where $|\lambda\rangle$ is the basis Fock vector labeled by a partition $\lambda$, then it is impossible.

In this context the work \cite{TakashiLMP} devoted to the study of Bruhat cells of '$A_\infty$ group' is helpful.

\section{Polynomials and vertex operators, a remark to \cite{paper2} \label{vertex}}

Polynomial BKP tau functions were studied in \cite{KvdL1},\cite{KMMMZ-2},\cite{KvdLRoz},\cite{KvdLRoz} and in 
\cite{paper2},\cite{paper3}.

In Example 4.4 in \cite{paper2} we considered polynomials denoted by
$s_\lambda\left(\tb|\tilde{A}^{\bf r}(\pb))\right)$
and $Q_\mu\left(\tfrac12\tb|{A}^{\bf r}(\pb_B))\right)$ there.
Below these polynomials will be denoted  $s_\lambda\left(\tb|\pb,{\bf r})\right)$ and 
$Q_\mu\left(\tb_{\rm B}|\pb_{\rm B},{\bf r})\right)$ respectively (see Remark \ref{notations}).

\paragraph{KP case.}
The following formula is known:
$$
s_\lambda(\tb+\pb)  =
\langle 0|\gamma(\tb)\gamma(\pb)|\lambda\rangle
=\sum_{\rho\in\Pa}\langle 0|\gamma(\tb)|\rho\rangle \langle \rho|\gamma(\pb)|\lambda\rangle=
\,\sum_{\rho \subseteq \lambda}
s_{\lambda/\rho}({\bf p})\,s_{\rho}({\bf t}),
$$

In \cite{paper2} we studied examples of  polynomial $\tau$-functions and introduced the generalized Schur function
parametrized in terms two sets of
 infinite parameters ${\bf r}:=\{r(j)\}_{j \in \Zb}$ and  ${\bf p}=(p_1,p_2,p_3,\dots)$:
\be
s_\lambda\left(\tb|\pb, {\bf r}\right)  =
\langle 0|\gamma(\tb)\gamma_r(\pb)|\lambda\rangle
=
\,\sum_{\rho \subseteq \lambda}
\,r_{\lambda/\rho}\,s_{\lambda/\rho}({\bf p})\,s_{\rho}({\bf t}),
\label{s_lambda_r_p_expansion}
\ee
where $s_{\lambda/\rho}({\bf p})$ is the skew Schur function corresponding to the skew
partition $\lambda/\rho$,
\be
r_{\lambda/\rho}:= \prod_{(i,j)\in \lambda/\rho } r(j-i)
\ee
is the content product defined on the skew diagramm $\lambda/\rho$ (that is the product over all
nodes with coordinates $(i,j)$ of the diagram, for instance $r_{(2,2)/(1)}=r(1)r(-1)r(0)$)
and where
\be
\gamma_r(\pb):=e^{\sum_{j=1}^\infty \frac 1j p_j \sum_{k\in\Zb} \psi_k \psi^\dag_{k+j} r(k+1)\cdots r(k+j)}
\ee
so that $\gamma(\pb)=\gamma_{r\equiv 1}(\pb)$. The example of such polynomials is
the Laguerre symmetric function ${\cal L}_\lambda$ introduced by Olshanski in \cite{Ol}
(see Definition 4.3 there). For this very case we choose
\bea\label{Olshanski}
\pb &\& =\tb_0:=(1,0,0,\dots)
\label{p_0_def}
\\
t_j &\&= t_j(\xb) = \sum_{k=1}^N x_k^j
\label{t-x}
\cr
r(j)&\&=r(j;z,z'):=-(z+j)(z'+j),
\label{r_j_laguerre}
\eea
eq.~(\ref{s_lambda_r_p_expansion}) takes the form
\be
{\cal L}_\lambda(\xb) =\sum_{\rho \subseteq \lambda} (-1)^{|\lambda|-|\rho|}
\,\frac{\dim \lambda/\rho}{(|\lambda|-|\rho|)!}\,s_{\rho}({\tb}({\xb}))
\, \prod_{(i,j)\in \lambda/\rho }(z+j-i)(z'+j-i)
\ee

Consider the vertex operator introduced in the context of classical integrable system in the series
articles by Kyoto school (see, for instance \cite{JM})
\footnote{At first the vertex operator was introduced in \cite{PogrebkovSushko} in different context.}
\be
V^\pm(z)=
e^{\pm\sum_{j=1}^\infty \frac 1jz^j p_j} z^{\pm N}e^{\mp \partial_N}
e^{\mp\sum_{j=1}^\infty  z^{-j} \frac{\partial}{\partial p_j}},\quad z\in S^1
\ee
Introduce
\be\label{A_KP}
A_m({\bf r})=\res_z
\vdots  \left(\left( \frac 1z r(D) \right)^{m} \cdot V^+(z) \right) V^{-}(z) \vdots \frac{dz}{z},\quad m\neq 0,
\ee
where $r(D)$ acts on functions on the circle as $r(D)\cdot z^n = r(n) z^n $ ($D=z\frac{\partial}{\partial z}$ is
the Euler operator on the circle). For instance
$\left( \frac 1z r(D) \right)^2 =\frac{1}{z^2}r(D-1)r(D) $. By $\vdots X \vdots$ we denote the bosonic normal
ordering, that means that all derivates $\frac{\partial}{\partial p_j}$ are moved to the right. For instance,
$\vdots \frac{\partial^2}{\partial p_i^2} p_k \frac{\partial}{\partial p_j} \vdots =  p_k
\frac{\partial^3}{\partial p_i^2\partial p_j}$.

In case $r \equiv 1$,
\be
A_m({\bf r})=\begin{cases}
              p_m \quad {\rm if}\,m >0\\
              m\frac{\partial}{\partial p_m}\,\,\,\quad {\rm if}\, m <0
             \end{cases}
\ee

We have the realization of the Heisenberg algebra $[A_m({\bf r}),A_n({\bf r})]=n\delta_{m,n}$, therefore
\be
e^{\sum_{m<0} \frac{p_m}{m}A_m({\bf r})}\cdot e^{\sum_{m>0} \frac{p_m}{m}A_m({\bf r})}=
e^{\sum_{m>0}\frac{p_m p_{-m}}{m}}
e^{\sum_{m>0} \frac{p_m}{m}A_m({\bf r})}\cdot e^{\sum_{m<0} \frac{p_m}{m}A_m({\bf r})}
\ee

\begin{proposition}\label{vertex-on-Schur}
\be
s_\lambda(\tb| \pb , {\bf r}) = e^{\sum_{m<0} \frac{p_m}{m}A_m({\bf r})}  \cdot s_{\lambda}(\tb)
\ee
\end{proposition}

In particular,
\be
 {\cal L}_{\lambda}(\xb)= \left( e^{ A_{-1}({\bf r})}  \cdot s_{\lambda}(\tb)\right)
 |_{\xb=\xb(\tb)}
\ee

\paragraph{BKP case.} It is quite similar in the BKP case.

Here we have

$$
Q_\alpha(\tb^{\rm B}+\pb^{\rm B})  =
\langle 0|\gamma^{\rm B}(2\tb^{\rm B})\gamma^{\rm B}(2\pb^{\rm B})\Phi_\alpha|0\rangle
=
\sum_{\theta\in\DP}2^{-\ell(\theta)}\langle 0|\gamma^{\rm B}(2\tb^{\rm B})\Phi_\theta|0\rangle
\langle 0|\Phi_{-\theta}\gamma^{\rm B}(2\pb^{\rm B})\Phi_\alpha|0\rangle
$$
$$
=\,\sum_{\theta \subseteq \alpha}2^{-\ell(\theta)}
Q_{\alpha/\theta}(\pb^{\rm B})\,Q_{\theta}(\tb^{\rm B}),
$$
where $Q_{\alpha/\theta}$ is the skew projective Schur function \cite{Mac}.

Following \cite{paper2} we introduce
\be
Q_\alpha(\tb^{\rm B}|\pb^{\rm B}, {\bf r}) :=
\langle 0|\gamma^{\rm B}(2\tb^{\rm B})\gamma_r(2\pb^{\rm B})\Phi_\alpha|0\rangle
=
\,\sum_{\theta \subseteq \alpha}
\,r^{\rm B}_{\lambda/\rho}\,Q_{\alpha/\theta}(\pb^{\rm B})\,Q_{\theta}(\tb^{\rm B})
\ee
where
 $\tb^{\rm B}=(t_1,t_3,t_5,\dots)$ and
 $\pb^{\rm B}=(p_1,p_3,p_5,\dots)$ are sets of parameters 
 and ${\bf r}=(r(1),r(2),r(3),\dots)$ is another
sets of parameters. Then
\be
\gamma_r^{\rm B}(2\pb^{\rm B}):=e^{\sum_{j=1,{\rm odd}}^\infty \frac 2j p_j
\sum_{k\in\Zb} (-1)^k\phi^\pm_{k-j} \phi^\pm_{-k} r(k+1)\cdots r(k+j)}
\ee
(so that $\gamma^{\rm B}(2\tb^{\rm B})=\gamma^{\rm B}_{r\equiv 1}(2\tb^{\rm B})$) and
\be
r^{\rm B}_{\alpha/\theta}:= \prod_{(i,j)\in \alpha/\theta } r(j)
\ee
where the product goes over all nodes $ (i, j) $ of the (skew) Young diagram, but each node $ (i, j) $ is 
assigned a number $ r (j) $, 
which depends only on the coordinate $ j $, but independent of the coordinate $ i $.
That is the '$B$-type' content product defined on the skew diagramm $\alpha/\theta$ (that is the product over all
nodes with coordinates $(i,j)$ of the diagram, for instance $r^{\rm B}_{(5,2)/(1)}=r(2)r(3)r(4)r(5)r(1)r(2)$).

\begin{proposition}\label{vertex-on-Q-Schur}
\be
Q_\alpha(\tb^{\rm B}|\pb^{\rm B}, {\bf r}) =
e^{2\sum_{m<0,{\rm odd}} \frac{p_m}{m}A^{\rm B}_m({\bf r})}  \cdot Q_{\alpha}(\tb^{\rm B})
\ee
where
\be\label{A_BKP}
A^{\rm B}_m({\bf r})=\frac 12 \res_z
\vdots  \left(\left( \frac 1z r(D) \right)^{m} \cdot V^{{\rm B}}(z) \right) V^{{\rm B}}(-z)
\vdots \frac{dz}{z},\quad m\,{\rm odd},
\ee
\end{proposition}

\br
In case $r(j)=r(1-j)$ we get
\be
e^{\sum_{m<0,{\rm odd}} \frac{p_m}{m}A_m({\bf r})}
=e^{2\sum_{m<0,{\rm odd}} \frac{p_m}{m}A^{\rm B}_m({\bf r})}
\ee
Then it follows that if we take (\ref{vertex-on-Schur}) and (\ref{vertex-on-Q-Schur})
as definitions of polynomials $s_\lambda(\tb|\pb, {\bf r})$
and $Q_\alpha(\tb^{\rm B}|\pb^{\rm B}, {\bf r})$
 we obtain from (\ref{sQQ}):

\begin{corollary} 
 \be
s_{(\alpha|\beta)}(\tb'|\pb', {\bf r})=\sum_{(\nu^+,\nu^-)\in \PP(\alpha,\beta)}
{a}^{\nu^+,\nu^-}_{\alpha,\beta}
Q_{\nu^+}(\tb^{\rm B}|\pb^{\rm B}, {\bf r})Q_{\nu^-}(\tb^{\rm B}|\pb^{\rm B}, {\bf r})
\ee
where $\tb'=(t_1,0,t_3,0,\dots)$, $\pb'=(p_1,0,p_3,0,\dots)$ and 
$\tb^{\rm B}=\frac12(t_1,t_3,t_5,\dots)$, $\pb^{\rm B}=\frac12(p_1,p_3,p_5,\dots)$.
\end{corollary}

In \cite{paper2} this equality was
derived fermionically.

\er

\subsection{Eigenproblem for polynomials \label{eigen} }

\paragraph{KP case.}
Let
\be
\mathbb{T} := e^{-\sum :\psi_i\psi_i^\dag: T_i}
 \equiv e^{ \res_z :\left(T(D)\cdot\psi(z) \right)\psi^\dag(z): \frac{dz}{z}}
\ee
where the sign $: A: $ denotes $ A- \langle 0 | A | 0 \rangle $ (the fermionic normal 
ordering of the expression quadratic in fermions),
where $T_i,\,i\in\mathbb{Z}$ is a set of numbers, where $D=z\frac{d}{dz}$
and where we imply that $T(D)\cdot z^k=T(k)z^k$.
Such diagonal operators were used in \cite{KMMM},\cite{NakatsuTakasaki},\cite{Kharchev} and also in
\cite{Takasaki-95},\cite{OS},\cite{HO-convolution},\cite{Wiegmann} in quite different contexts.
We have
\be
\langle 0|\gamma(\pb) \mathbb{T}|\lambda\rangle = s_\lambda(\pb)e^{-\sum_{i=1}^{\ell(\lambda)} T_{h_i}}
\ee

Take  $T_i=i^2$. The bosonozed version in this case is as follows:
\be
\mathbb{T}^{Bos}=e^{\sum_{a+b=c}
\left(p_a p_b\frac{\partial}{\partial p_c} + ab p_c\frac{\partial^2}{\partial_a\partial_b} \right)}
\ee
(this operator is also known as the cut-and-join operator introduced in \cite{GJ} in different context).

Introduce
\be
\mathbb{T}(A):=e^{\sum_{m<0} \frac 1m p_mA^{Fer}_m} \mathbb{T} e^{-\sum_{m<0} \frac 1m p_mA^{Fer}_m}
\ee
where $A^{Fer}_m,\,m=\pm 1,\pm 2,\pm 3,\dots $ are the fermionic counterparts of $A_m,\,m=\pm 1,\pm 2,\pm 3,\dots $
introduced by (\ref{A}), i.e.
\be\label{A-F}
A^{Fer}_m=\res_z
:  \left(\left( \frac 1z r(D) \right)^{m} \cdot \psi(z) \right) \psi^{\dag}(z) : \frac{dz}{z},
\ee
(such operators were introduced in \cite{OS}).

Then from fermionic expression for $s_\lambda(\tb)$ by the bosonization rule we have
\be
\mathbb{T}^{Bos}\cdot s_\lambda(\tb) = s_\lambda(\tb)e^{-\sum_{i=1}^{\ell(\lambda)} T_{h_i}}
\ee
where $\ell(\lambda)$ is the length of $\lambda$ (the number of nonvanishing parts of
$\lambda$ and where
\be
h_i = \lambda_i-i+n
\ee
\be
 e^{{\cal H}^{Bos} }:=
 e^{\sum_{m<0} \frac{p_m}{m}A_m({\bf r})} \mathbb{T}^{Bos} e^{-\sum_{m<0} \frac{p_m}{m}A_m({\bf r})}
\ee

From Proposition \ref{vertex-on-Schur} it follows that
\be
 e^{t{\cal H}^{Bos}} \cdot
  s_\lambda(\tb,\pb | {\bf r})= e^{t E(h_1,\dots,h_N)} s_\lambda(\tb,\pb | {\bf r})
\ee
where
\be
E(h_1,\dots,h_N) = -\sum_{i=1}^{\ell(\lambda)} T_{h_i}
\ee

\paragraph{BKP case.}
Here
\be
\mathbb{T} := e^{-\sum_{i>0} :\phi_i\phi_i^\dag: T_i}
 \equiv e^{ \res_z :\left(T(D)\cdot\phi(z) \right)\phi(-z): \frac{dz}{z}}
\ee
where $T_i,\,i>0$ is a set of numbers.
Such diagonal operators were used in \cite{Or}.
We have
\be
\langle 0|\gamma^{\rm B}(2\pb^{\rm B}) \mathbb{T} \Phi_\alpha|0\rangle = 
Q_\alpha(\pb^{\rm B})e^{-\sum_{i=1}^{\ell(\alpha)} T_{\alpha_i}}
\ee

Example: $T_i=i^3$.
The bosonization of
\be
\frac 12 \sum_{j\in\mathbb{Z}}  j^3(-1)^j:\phi_{j}\phi_{-j}: =
\res_z :\left( \left(z\frac{\partial}{\partial z}\right)^3
\cdot \phi(z)\right) \phi(-z):\frac {dz}{z}
\ee
gives \cite{MMNO}
\bea
\frac 12 \sum_{n>0} n^3 p_n\partial_n +
\frac {1}{2}\sum_{n>0} n p_n\partial_n
+4\sum_{n_1,n_2,n_3\,{\rm odd}} p_{n_1}p_{n_2}p_{n_3}  (n_1+n_2+n_3)\partial_{n_1+n_2+n_3}
\\
+3\sum_{n_1+n_2=n_3+n_4\,{\rm odd}} p_{n_1}p_{n_2} n_3n_4\partial_{n_3}\partial_{n_4} +
\sum_{n_1,n_2,n_3\,{\rm odd}} p_{n_1+n_2+n_3}\partial_{n_1}\partial_{n_2}\partial_{n_3}
\eea

Introduce
\be
\mathbb{T}(A):=e^{\sum_{m<0} \frac 1m p_mA^{Fer{\rm B}}_m} \mathbb{T} e^{-\sum_{m<0} \frac 1m p_mA^{Fer{\rm B}}_m}
\ee
where $A^{Fer{\rm B}}_m,\,m=\pm 1,\pm 2,\pm 3,\dots $ are the fermionic counterparts 
of $A^{\rm B}_m,\,m=\pm 1,\pm 2,\pm 3,\dots $
introduced by (\ref{A_BKP}), i.e.
\be\label{A-FB-B}
A^{Fer{\rm B}}_m=\frac12 \res_z
:  \left(\left( \frac 1z r(D) \right)^{m} \cdot \phi(z) \right) \phi(-z) : \frac{dz}{z},
\ee
(such operators were introduced in \cite{Or}).

Then from fermionic expression for $Q_\alpha(\pb^{\rm B})$ by the bosonization rule we have
\be
\mathbb{T}^{Bos}\cdot Q_\alpha(\tb^{\rm B}) = Q_\alpha(\tb^{\rm B})e^{-\sum_{i=1}^{\ell(\alpha)} T_{\alpha_i}}
\ee
where $\ell(\alpha)$ is the length of $\alpha$. Then
\be
 e^{{\cal H}^{Bos{\rm B}} }:=
 e^{\sum_{m<0} \frac{p_m}{m}A^{\rm B}_m({\bf r})} \mathbb{T}^{Bos} e^{-\sum_{m<0} \frac{p_m}{m}A^{\rm B}_m({\bf r})}
\ee

From Proposition \ref{vertex-on-Q-Schur} it follows that
\be
 e^{t{\cal H}^{Bos{\rm B}}} \cdot
  Q_\alpha(\tb^{\rm B},\pb^{\rm B} | {\bf r})= 
  e^{t E(\alpha_1,\dots,\alpha_N)} Q_\alpha(\tb^{\rm B},\pb^{\rm B} | {\bf r})
\ee
where
\be
E(\alpha_1,\dots,\alpha_N) = -\sum_{i=1}^{\ell(\alpha)} T_{\alpha_i}
\ee

\section{A note to \cite{MMNO} about $W_{1+\infty}$ and $BW_{1+\infty}$ and certain matrix models \label{W}}

The algebra of the differential operators on the circle of type $B$ denoted $BW_{1+\infty}$ with the central extention
was considered in  \cite{Leur1994},\cite{Leur1996}.

\paragraph{BKP case.}
The relation of the hypergeometric series
which depends on two sets $\pb=(p_1,p_3,\dots)$ and $\tilde{\pb}=(\tilde{p}_1,\tilde{p}_3,\dots)$
to the action of the exponentials of $BW_{1+\infty}$:
\be\label{hyp=BW}
\sum_{\alpha\in\DP} 2^{-\ell(\alpha)} r^{\rm B}_\alpha Q_\alpha(\pb^{\rm B})Q_\alpha(\tilde{\pb}^{\rm B})
=
e^{2\sum_{m>0,{\rm odd}} \frac{\tilde{p}_m}{m}A^{\rm B}_m({\bf r})}  \cdot 1
\ee
was derived in \cite{Or}. Here the commutative set of $A^{\rm B}_m({\bf r})\in BW_{1+\infty} $ is given by (\ref{A_BKP}) 
 and $\tilde{\pb}=(\tilde{p}_1,\tilde{p}_3,\dots)$ play the role of group parameters.
This is the direct analogue of the similar relation in the KP case \cite{OS2000},\cite{OS}.

For example, take $r\equiv 1$, in this simplest case $A^{\rm B}_m({\bf r})=p_m$ which does not contain differential operators,
and (\ref{hyp=BW}) takes the form
\be\label{hyp=BWvac}
\sum_{\alpha\in\DP} 2^{-\ell(\alpha)}  Q_\alpha(\pb^{\rm B})Q_\alpha(\tilde{\pb}^{\rm B})
=
e^{2\sum_{m>0,{\rm odd}} \frac{\tilde{p}_m p_m}{m}}
\ee

$\,$

It is interesting that the similar formula was obtained  for the different series which is applicibale for
a special matrix integral (see \cite{MMNO}):
\begin{proposition}
Suppose $\tilde{p}_m=\tr 
\Lambda^{-m}$. We have
\be
\sum_{\alpha\in\DP} 2^{-\ell(\alpha)} Q_{2\alpha}(\pb^{\rm B})Q_\alpha(\tilde{\pb}^{\rm B})
\prod_{i=1}^{\ell(\alpha)}(2\alpha_1-1)!!
=
 e^{\sum_{i>0,{\rm odd}} \frac 2r \tilde{p}_r L^{\rm B}_{-1}(r,\pb^{\rm B})}\cdot 1
\ee
\be\label{integral}
= C(\Lambda)\int e^{-\tr\left(X^2\Lambda\right)+\sum_{i>0,\,{\rm odd}}\frac 1i p_i \tr X^i }dX
\ee
\end{proposition}
(This integral is different from the so-called generalized Kontsevich matrix model 
\cite{KMMMZ-1},\cite{KMMMZ-2},\cite{GKM}) which reads as
$$
C(\Lambda)\int e^{-\tr\left(\Lambda \sum_{i>0} p_i \tr X^i \right)+\sum_{i>0}\frac 1i p_i \tr X^i }dX
$$
As for
the Kontsevich integral it is obtained by putting $p_m=\delta_{m,3}$ in (\ref{integral}) and this integral will be 
denoted $Z(\Lambda)$.)
The notations are as follows:
\begin{itemize}
 \item  $X$ is a Hermitian matrix,
\be\label{dX}
dX=\prod_{i>j}d\Im X_{ij}\prod_{i\ge j} d\Re X_{ij},
\ee
 $\Lambda$ is a matrix with positive entries, $C(\Lambda)$ is the noramlization constant
 \item $\pb^{\rm B}=( {p}_1,{p}_3,{p}_5,\dots)$ is the set of free parameters (coupling
 constants)
 \item
$\tilde \pb^{\rm B}=(\tilde {p}_1,\tilde{p}_3,\tilde{p}_5,\dots)$, where
$\tilde{p}_m=\tr\Lambda^{-m}$
 \item $L_{-1}(r,\pb^{\rm B}),\,r=1,3,5,\dots$ is the set of commuting operators
 \be\label{L^B_(-1)-based-seriesB}
L^{\rm B}_{-1}(r,\pb^{\rm B})=
\frac 12 \res_z \vdots\left( \frac{d^r (z^{-1}V^{\rm B}(z))}{dz^r}  \right)V^{\rm B}(-z)\vdots\frac{dz}{z}
\ee
which starts with the element of the BKP Virasoro algebra (which is the restriction of the KP Virasoro
symmetries \cite{Orlov1988},\cite{GO}
to the BKP case \cite{Leur1994}) which is responsible for the Galileo transformation
of the BKP hierarchy:
\be
L^{\rm B}_{-1}(1,\pb^{\rm B})=p_1^2+\sum_{m>0,{\rm odd}} m p_{m+2} \frac{\partial}{\partial p_{m}}
\ee
\end{itemize}
The restriction $p_m=\delta_{m,3}$ result in the famous Kontsevich model. The nice form of the perturbation
series of this model in terms of the projective Schur functions was found in \cite{MMq}.

\paragraph{KP case}

It should be compared with the relation in the KP case concerning well-known two- and one-matrix models:
\be\label{2MMseries}
\sum_{\lambda\in\Pa} s_{\lambda}(\pb)s_\lambda(\tilde{\pb})
\prod_{(i,j)\in\lambda} (N+j-i)
=
 e^{\sum_{i>0} \frac 1r \tilde{p}_r L_{-1}(r,\pb)}\cdot 1
\ee
\be\label{2MM}
= C\int e^{-\tr ( XY )+\sum_{i>0}\frac 1i p_i \tr X^i + \sum_{i>0}\frac 1i \tilde{p}_i \tr Y^i }dXdY
\ee
The notations are as follows:
\begin{itemize}
 \item  $X,Y$ are Hermitian matrices,
\be\label{dXdY}
dX=\prod_{i>j}d\Im X_{ij}\prod_{i\ge j} d\Re X_{ij},\quad
dY=\prod_{i>j}d\Im Y_{ij}\prod_{i\ge j} d\Re Y_{ij}
\ee
 $\Lambda$ is a matrix with positive entries, $C(\Lambda)$ is the normalization constant
 \item $\pb=( {p}_1,{p}_2,{p}_3,\dots)$,$\tilde \pb=(\tilde {p}_1,\tilde{p}_2,\tilde{p}_3,\dots)$,
 are the sets of free parameters (coupling constants)
 \item $L_{-1}(r,\pb),\,r=1,3,5,\dots$ is the set of commuting operators
 \be\label{L_(-1)-based-series}
L_{-1}(r,\pb)=
 \res_z \vdots\left( \frac{d^r V^{+}(z))}{dz^r}  \right)V^{-}(z)\vdots \frac{dz}{z}
\ee
which starts with the element of the KP Virasoro algebra \cite{Orlov1988},\cite{GO} which is
responsible for the Galileo transformation
of the KP hierarchy:
\be
L_{-1}(1,\pb)=p_1+\sum_{m>0,{\rm odd}} m p_{m+1} \frac{\partial}{\partial p_{m}}
\ee
\end{itemize}

The pertubation series for the famous one-matrix model is obtained by the restriction of (\ref{2MM})
$\tilde{p}_m=\delta_{m,2}$ (see Section 4.4 in \cite{Orlov2002Acta}, and \cite{Orlov2002Tau},\cite{OH2002}).

\section{An evaluation of the Schur and the projective Schur functions
at special sets of power sums $\pb=\pb[{{\cal R}}]$.
\label{evaluation}}

Consider a given set of parameters 
$\pb=(p_1,p_2,p_3,\dots )$. This set gives rise to the sets
\bea
\pb':&=&(p_1,0,p_2,0,p_3,\dots)\\
\pb^{\rm B}:&=&(p_1,p_3,p_5,\dots)
\eea
and to the sets
\bea \label{pb[r]}
\pb[{\cal R}]:&=&
(\underbrace{0,\dots,0,p_1}_{\cal R},\underbrace{0,\dots,p_2}_{\cal R},\underbrace{0,\dots,0,p_3}_{\cal R},
\underbrace{0,\dots,p_4}_{\cal R},
\underbrace{0,\dots,0,p_5}_{\cal R},\dots )\\
 \label{pb'[r]}
\pb'[{\cal R}]:&=&
(\underbrace{0,\dots,0,p_1}_{\cal R},\underbrace{0,\dots,0}_{\cal R},\underbrace{0,\dots,0,p_3}_{\cal R},
\underbrace{0,\dots,0}_{\cal R},
\underbrace{0,\dots,0,p_5}_{\cal R},\dots )\\
 \label{pb^B[r]}
\pb^{\rm B}[{\cal R}]:&=&
(\underbrace{0,\dots,0,p_1}_{\cal R},\underbrace{0,\dots,0,p_3}_{\cal R},
\underbrace{0,\dots,0,p_5}_{\cal R},\dots )
\eea
where in (\ref{pb'[r]}) and in (\ref{pb^B[r]}) we put ${\cal R}>1,\,{\rm odd}$.   In other words:
  \be \label{t[r]-t}
  \frac 1k p_k[{\cal R}]=\frac 1j p_j \delta_{k,j{\cal R}},\quad {\rm both}\quad {\cal R},j\,\,{\rm odd}
  \ee
where $k=1,2,3,\dots$, $j=1,3,5,\dots$ and ${\cal R}$ is a given odd number.

We will study Schur functions which depend on power sums specified by (\ref{pb[r]}) and by (\ref{pb'[r]})
and  will study projective Schur functions which depend on power sums specified by (\ref{pb^B[r]})
  
  In some problems \cite{Alex2},\cite{MMNO} it is important to evaluate symmetric functions which
  depends on such sets of power sums.

\subsection{Schur functions}

Let us use the trick applied in \cite{Orlov2002Acta}.
Using ${\cal R}$-component fermions obtained by the re-enumeration of
$\psi_m,\psi^\dag_m$: $\psi^{(j)}_i=\psi_{i{\cal R}+j}$  and $\psi^{(j)\dag}_i=\psi^\dag_{i{\cal R}+j}$,
$j=0,\dots,{\cal R}-1$
we obtain that $J^{(0)}_m+\cdots +J^{({\cal R}-1)}_m=J_m$ where the related currents are
$J^{(j)}_m=\sum_{i} \psi^{(j)}_{i}\psi^{(j)\dag}_{i+m}$, see \cite{JM}.
We shall call $j$ the {\em color} of fermion modes $\psi^{(j)}_i,\psi^{(j)\dag}_i$ and also of the modes
$\psi_{i{\cal R}+j},\psi^\dag_{i{\cal R}+j}$ and of the current $J^{(j)}_i$ for each $i\in\mathbb{Z}$.
Next, each $\Psi_{\alpha,\beta}$ converts to the product
$(-1)^\omega\Psi^{(0)}_{\lambda^{0}}\cdots \Psi^{({\cal R}-1)}_{\lambda^{{\cal R}-1}} $, where $\omega$ is the sign of
the permutation of the Fermi components. Let us introduce
\be
\gamma^{(0)}(\pb)\cdots \gamma^{({\cal R}-1)}(\pb)=\gamma(\pb[{\cal R}]),\quad
\gamma^{(j)}(\pb)=e^{\sum_{m>0} \frac{p_m}{m}J^{(j)}_m},\quad j=0,\dots, {\cal R}-1
\ee

Let us denote the ordered set of ${\cal R}$ partitions $\lambda^i,\, i=0,\dots, {\cal R}-1$ as  
$\overrightarrow{\lambda}=(\lambda^0,\dots,\lambda^{{\cal R} -1 })$ and the ordered set of ${\cal R}$ 
integers $n^{(i)},\,i=0,\dots,{\cal R}-1$
as $\overrightarrow{n}=(n^{(0)},\dots,n^{({\cal R}-1)})$, 
$|\overrightarrow{n}|:=n^{(0)}+\cdots + n^{({\cal R}-1)}$. 

As one can see the action of $\gamma(\pb[{\cal R}])$ does not mix fermions and partitions
of different color.

From the fermionic picture one can see that 
\bl\label{}
Suppose ${\cal R}$, $\lambda$ and $n$ are given. Then there exist the set
$\lambda^i=\lambda^i(\lambda,n)$ and the set $n^{(i)}=n^{(i)}(\lambda,n)$,  $i=0,\dots,{\cal R}-1$ and an 
integer $\omega(\overrightarrow{\lambda},\overrightarrow{n})$ such that
\bea\label{lambda=lambda...lambda}
|\lambda;n\rangle = 
(-1)^{\omega(\overrightarrow{\lambda},\overrightarrow{n})}|(\lambda^{0};n^{(0)}),
\dots,(\lambda^{{\cal R}-1};n^{({\cal R}-1)})\rangle
=(-1)^{\omega(\overrightarrow{\lambda},\overrightarrow{n})} |\overrightarrow{\lambda};\overrightarrow{n}  \rangle
\\
\label{lambda=lambda...lambda*}
\langle\lambda;n | = (-1)^{\omega(\overrightarrow{\lambda},\overrightarrow{n})}
\langle (\lambda^{{\cal R}-1};n^{({\cal R})}),\dots,(\lambda^{0};n^{(0)}) |
=(-1)^{\omega(\overrightarrow{\lambda},\overrightarrow{n})} \langle \overrightarrow{\lambda};\overrightarrow{n} |
\eea
and
\be\label{n=n+...+n}
n^{(0)}+\cdots +n^{({\cal R}-1)}=n
\ee
The converse is also true: for any given set of partitions $\lambda^{0},\dots,\lambda^{{\cal R}-1}$
and for any given set of numbers $n^{(0)},\dots,n^{({\cal R}-1)}$ satisfying (\ref{n=n+...+n}), there exist 
$\lambda=\lambda( \overrightarrow{\lambda},\overrightarrow{n} )$ and 
$\omega(\overrightarrow{\lambda},\overrightarrow{n})$ that satisfies (\ref{lambda=lambda...lambda}),
(\ref{lambda=lambda...lambda*}).

\el
The integer $\omega( \overrightarrow{\lambda},\overrightarrow{n} )$ is determined by how fermions of 
different colors alternate in the state $|\lambda;n\rangle$.

We have $\langle \overrightarrow{\lambda};\overrightarrow{n} | \overrightarrow{\mu};\overrightarrow{m}  \rangle
=\delta_{\overrightarrow{n},\overrightarrow{m}} \delta_{\overrightarrow{\lambda},\overrightarrow{\mu}}$.

\bl For a given $\lambda$ and a given $n$ there exist such
$\overrightarrow{n}=(n^{(0)},\dots,n^{({\cal R}-1)})$ 
where ${\cal R}=\ell(\lambda)+\ell(\lambda^{\rm tr})$ and
where $|\overrightarrow{n}|=n$
that the state $|\lambda;n\rangle$ 
can be considered as the multicharged vacuum state $|n^{(0)},\dots,n^{({\cal R}-1)}\rangle$:
\be
| \lambda;n\rangle=(-1)^{\omega(\overrightarrow{0},\overrightarrow{n})} |
\overrightarrow{0},\overrightarrow{n} \rangle
\ee
where $|\overrightarrow{n}|=n$. 

\el
For example $|(1)\rangle =\psi_0\psi^\dag_{-1}|0\rangle =\psi^{(0)}_0\psi^{\dag(1)}_{-1}|0\rangle=
|n^{(0)},n^{(1)},n^{(2)}\rangle$ 
where $n^{(0)}=1$, $n^{(1)}=-1$ and $n^{(2)}=0$.

Denote
\be
s_{\overrightarrow{\lambda}}(\pb) :=s_{\lambda^{0}}(\pb)\cdots s_{\lambda^{{\cal R}-1}}(\pb)
\ee
We get
\begin{proposition}
For a given $\overrightarrow{n}=\left( n^{(0)},\dots,n^{({\cal R}-1)} \right)$ we have
 \be
 \langle \overrightarrow{n}|   \gamma(\pb[{\cal R}])=
\sum_{\overrightarrow{\lambda} \in \Pa^{\times {\cal R}}}
 s_{\overrightarrow{\lambda}}(\pb)
 \langle \overrightarrow{\lambda};\overrightarrow{n} |
 \ee
 or, the same
 \bea
 s_{\overrightarrow{\lambda}}(\pb)  &=& 
 \langle \overrightarrow{n}|   \gamma(\pb[{\cal R}])|\overrightarrow{\lambda};\overrightarrow{n}\rangle =
 \\
  (-1)^{\omega(\overrightarrow{0},\overrightarrow{n}) +\omega(\overrightarrow{\lambda},\overrightarrow{n})}
 \langle \mu | \gamma(\pb[{\cal R}])|\lambda\rangle 
 &=& (-1)^{\omega(\overrightarrow{0},\overrightarrow{n}) +\omega(\overrightarrow{\lambda},\overrightarrow{n})}
 s_{\lambda/\mu}
 \eea
 where $\mu =\lambda(\overrightarrow{0},\overrightarrow{n})$.
\end{proposition}
Taking 
\be\label{vect-n=0}
\overrightarrow{n}=0
\ee
we obtain
\begin{corollary}

\be\label{s=s...s}
s_\lambda(\pb[{\cal R}])= (-1)^\omega s_{\lambda^{0}}(\pb[1])\cdots s_{\lambda^{{\cal R}-1}}(\pb[1])
\ee
where $\pb[1]=\pb=(p_1,p_2,\dots)$ and
\be
s_{\lambda^{j}}(\pb)=\langle 0|\gamma^{(j)}(\pb)\Psi^{j}_{\lambda^j}|0\rangle
\ee
\end{corollary}
Relation (\ref{s=s...s}) was earlier used in \cite{Orlov2002Acta},\cite{OH2002},\cite{OrlovTMP-rational-soliton}.

\paragraph{The case where $\lambda$ is the double of a strict partition $\alpha$: $\lambda=D(\alpha)$.}

The double of a strict partition $\alpha=(\alpha_1,\dots,\alpha_k)$ is the partition 
$\lambda=\left(\alpha_1,\dots,\alpha_k|I(\alpha_1),\dots,I(\alpha_k)\right)$, where $I(\alpha_i):=\alpha_i-1,\,i=1,\dots,k$.

In what follows we consider the case (\ref{vect-n=0}).

Let us choose the value of the odd number ${\cal R}$.
The set of numbers $\alpha_1,\dots,\alpha_k$ can be splitted into ${\cal R}$ subsets labelled 
by $c=0,1,\dots,{\cal R}-1$: 
each of which can be written
with the help of a strict partion denoted resepectively as $\alpha^c=(\alpha^c_1,\dots,\alpha_{\ell(\alpha^c)})$
where $c=0,1,\dots, {\cal R}-1$ and $\ell(\alpha^c)$ is the length of $\alpha^c$ as follows. 
Each $\alpha_i$ can be presented as ${\cal R}\alpha^c_j+c$ for 
a certain $j$ and $c$. That to say that $c$ is the fractional part of the number $\alpha_i$ after didvision
on $\cal R$. Let us call the number $c$ the {\em color} and the partition $\alpha^c$ the {\em colored partition}.
There is the subset labeled by $c=0$ which consists of ${\cal R}\alpha^0_i$, $i=1.\dots,$ and
 $\frac12 ({\cal R}-1)$  pairs of subsets labeled by $c$ and ${\cal R}-c$ where $c< \frac12 ({\cal R}-1)$.
 The colors $c$ and ${\cal R}-c$ we call complementary. One can note that $\alpha_i $ and $I(\alpha_i)$ are related to
 the complementary colors which are different for $c\neq 0$.
One can see that if $\lambda$ is the double of $\alpha$, then the colored partitions labeled by $c$ and by 
${\cal R}-c$ coincide: $\lambda^c=\lambda^{{\cal R}-c}$.

  We get

\begin{proposition}
 \be\label{S-S}
s_{(D(\alpha)}(\pb[{\cal R}])= s_{D(\alpha^0)}(\pb)\prod_{c=1}^{\frac 12 ({\cal R}-1)}
\left( s_{(\alpha^c|\beta^c)}(\pb)\right)^2
\ee
\end{proposition}

The projective Schur function $Q_\alpha(\pb[{\cal R}])$ is a square root of the 
Schur function $s_{(D(\alpha)}(\pb[{\cal R}])$.
Following \cite{MMNO} we derive it below. 

\subsection{Projective Schur functions \cite{MMNO}}

This piece is a result of the discussion of the interesting explicit conjecture of Alexandrov \cite{Alex}
(before it also suggested by Mironov and Morozov in an implicit form \cite{MM-Dec10} to the author in
connection with the work \cite{MMq} about the Kontsevich matrix model).
The conjecture is that the Kontsevich matrix model can be written as an example of the BKP hypergeometric
tau function and in this way is a certain generating function  of the 
so-called spin Hurwitz numbers \cite{MMN2019},\cite{Lee2018}.
We study this problem from the fermionic point of view in \cite{MMNO}. Our result was presented in
\cite{MMNO} and I copy it below to compare this result with the previous section.

We want to evaluate
\be\label{Q_R}
Q_{R}(\pb^{\rm B}[{\cal R}])=2^{\tfrac L2}\langle 0|\gamma(2\pb^{\rm B}[{\cal R}])
\phi_{R_1}\cdots \phi_{R_L} |0\rangle
\ee
where we suppose  $L$ to be even and where $\pb^{\rm B}[{\cal R}]$ is defined in (\ref{pb^B[r]}),
and
\be\label{Q_{NR}}
Q_{NR}(\pb^{\rm B}[{\cal R}])=2^{\tfrac L2}\langle 0|\gamma(2\pb^{\rm B}[{\cal R}])
\phi_{NR_1}\cdots \phi_{NR_L} |0\rangle
\ee

Steps I and II below deal with (\ref{Q_R}). Steps III and IV consider the link between (\ref{Q_R}) and
(\ref{Q_{NR}}).

\paragraph{Step 0}
 A given set $\pb^{\rm B}[{\cal R}]$, we use the notation $\pb^{\rm B}=(p_1,p_3,\dots)$ for 
 the related set by (\ref{pb^B[r]}).
 We also introduce $\pb':=(p_1,0,p_3,0,\dots)$.
 Using the canonical anticommutation relation
 \be
 [\phi_i,\phi_j]_+=(-1)^j \delta_{i+j,0}
 \ee
 we have
 \be
 [J_m , \phi_i] = \phi_{i-m},\quad m\,{\rm odd}
 \ee
 and
 \be\label{dynamics}
 \phi_j(\pb^{\rm B}[{\cal R}]):=e^{\sum_{m>0,{\rm odd}} \frac{2}{n{\cal R}}J_{m{\cal R}}p^{\rm B}_{m{\cal R}}[{\cal R}]}\phi_j
 e^{-\sum_{m>0,{\rm odd}} \frac{2}{n{\cal R}} J_{m{\cal R}}p_{m{\cal R}}[{\cal R}]}=
 \sum_{m\ge 0}\phi_{j-m{\cal R}}h_m(2\pb')
 \ee
 where $h_i$ are complete symmetric functions restricted on the set of odd labeled times:
\be\label{h}
e^{\sum_{n>0, {\rm odd}} \frac 2n p_n z^n} =\sum_{n\ge 0} z^n h_n(2\pb'), \quad h_n(2\pb') := s_{(n)}(2\pb'),
\ee
Let us note that the exponential (\ref{h}) is also the generating function for the elementary projective
Schur functions $q_n=Q_{n,0}$, thus
\be\label{h=q}
h_n(2\pb')=Q_{(n,0)}(\pb^{\rm B})
\ee

\paragraph{Step I}

For the evaluation of VEV in the r h s of (\ref{Q_R}) we shall use the Wick theorem. To do it
we need to consider the pairwise VEV:
 \be\label{pair-corr}
\langle 0|\gamma(2\pb^{\rm B}[{\cal R}]) \phi_{\beta}\phi_{\alpha} |0\rangle =
\langle 0| \phi_{\beta}(\pb^{\rm B}[{\cal R}])\phi_{\alpha}(\pb^{\rm B}[{\cal R}]) |0\rangle
\ee
\be
 =\langle 0|  \sum_{m\ge 0}\phi_{\beta-m{\cal R}}h_m(2\pb') \sum_{n\ge 0}\phi_{\alpha-n{\cal R}}h_n(2\pb')  |0\rangle
 \ee
 where according to the canonical pairing
\be\label{i-icorr}
\langle 0| \phi_{i}\phi_j |0\rangle =(-1)^i\delta_{i+j,0},\, j>0,
\ee
 contribute only the terms where
 \be\label{condition}
 \alpha - n{\cal R} + \beta -m{\cal R} =0
 \ee

As we see condition (\ref{condition}) implies that
the parts of $R$ consist of 3 groups:
 \begin{itemize}
  \item parts ${\cal R}\alpha^0$ that are divisible by ${\cal R}$
  \item parts presented as ${\cal R}\alpha^c+c$ where $c=1,\dots, \frac 12 ({\cal R}-1)$
  \item parts presented as ${\cal R}\beta^c+{\cal R}-c$ where $c=1,\dots, \frac 12 ({\cal R}-1)$
 \end{itemize}
 and nothing else. For a given $c$ we call the set of parts with such remainder of division by ${\cal R}$
 and the set of parts with remainder ${\cal R}-c$ {\it complementary}.

 Up to (\ref{Q-S}) $\alpha^c,\beta^c$ will denote numbers, not partitions.

 Consider a given $c$.
 Among the set of $\{R_i\}$ choose a pair of fermion with complementary labels:
 \be\label{alpha-beta}
 \alpha={\cal R}\alpha^c +c, \quad \beta= {\cal R}\beta^c+{\cal R}-c ={\cal R}(\beta^c+1) -c
 \ee
 \br
 We can say that there exists a part $R_x$ such that $\alpha^c=[R_x {\cal R}^{-1}]$ is the integer part of
 $R_x{\cal R}^{-1}$ and whose fractial part is $c=\{R_x{\cal R}^{-1}\}$, and there exists another part $R_y$ whose
 $[R_y{\cal R}^{-1}]$ is $\beta+1$ and whose $\{R_y{\cal R}^{-1}\}$ is $-c$. In this case
 $\langle 0|\gamma(2\pb^{\rm B}[{\cal R}])\phi_{R_x}\phi_{R_y}|0\rangle \neq 0$.
 \er
  We have
 \be\label{a+b}
 \alpha+\beta={\cal R} \alpha^c+{\cal R}\left( \beta^c +1\right)
 \ee
 and together with (\ref{condition}) we have the sum of nonvanishing terms $h_m(2\pb')h_n(2\pb')$ for which
 \be\label{n+m}
 n+m = \alpha^c + \beta^c +1
 \ee

 We obtain that (\ref{pair-corr}) is equal to
 \be
 \langle 0| \phi_{-c} \phi_{c} |0\rangle h_{\beta^c+1}(2\pb')h_{\alpha^c }(2\pb')+ \cdots +
 \langle 0| \phi_{-\alpha} \phi_{\alpha} |0\rangle h_{\alpha^c + \beta^c +1}(2\pb')h_{0}(2\pb')
 \ee
 \be\label{sum-hh}
 = (-1)^c h_{\beta^c+1}(2\pb')h_{\alpha^c}(2\pb')+ \cdots +(-1)^{\alpha +1} h_{\alpha^c + \beta^c }(2\pb')h_{1}(2\pb')
 + (-1)^\alpha h_{\alpha^c + \beta^c +1}(2\pb')h_{0}(2\pb')
 \ee
 For instance, for ${\cal R}=3$ (so we have the single $c=1 =\frac 12 ({\cal R}-1)$), take $\alpha^1=\alpha^2=0$,
 that yields
 $\alpha = 0+1=1$, $\beta=0+2=2$. We obtain
 \be
 \langle 0| \phi_2(\pb^{\rm B}[3])\phi_1(\pb^{\rm B}[3]) |0\rangle =\langle 0| \phi_{-1}\phi_1 |0\rangle h_1(2\pb')h_0(2\pb')
 =(-1)^1  h_1(2\pb')h_0(2\pb') = -2 p_1
 \ee

 Now notice that thanks to the fact that ${\cal R},j$ are odd, we have from (\ref{h})
 \be\label{h(odd-times)}
 h_i(2\pb')=(-1)^i h_i(-2\pb')
 \ee
therefore (\ref{sum-hh}) is written as
\be\label{sum-hh2}
  (-1)^{c+\alpha^c} h_{\beta^c+1}(2\pb')h_{\alpha^c}(-2\pb')+ \cdots +
  (-1)^{\alpha} h_{\alpha^c + \beta^c }(2\pb')h_{1}(-2\pb')
 + (-1)^\alpha h_{\alpha^c + \beta^c +1}(2\pb')h_{0}(-2\pb')
 \ee
 (the parity of $\alpha={\cal R}\alpha^c +c$ is equal to the parity of $c+\alpha^c$ because ${\cal R}$ is odd).

 We compare (\ref{sum-hh2}) with one-hook Schur function \cite{Mac}
\be
s_{(j|k)}(\pb) =
(-1)^{k} \sum_{i=0}^{k} h_{j+i+1}(\pb) h_{k-i}(-\pb).
\label{s_hook_bilinear}
\ee
where we take $\pb=2\pb'$.
And obtain that
\be
\langle 0|\gamma(2\pb^{\rm B}[{\cal R}]) \phi_{{\cal R}\beta^c+c^*}\phi_{{\cal R}\alpha^c+c} |0\rangle =
(-1)^{\beta^c+{\cal R}\alpha^c+c} s_{(\alpha^c|\beta^c)}(2\pb')=
(-1)^{\beta^c+\alpha^c+c} s_{(\alpha^c|\beta^c)}(2\pb')
\ee

Next, in the similar way consider
 parts didivisble by ${\cal R}$, the collection of such parts we denote ${\cal R}\mu$,
 $\mu=(\mu_1,\dots,\mu_\ell)$, $\mu_1>\dots>\mu_k\ge 0$, $k$ even.

 For a pairwise VEV using (\ref{dynamics}) in the quite similar way as before using (\ref{h=q})
 we get
 \be
\langle 0|\phi_{{\cal R}\mu}(\pb^{\rm B}[{\cal R}])\phi_{{\cal R}\nu}([\pb^{\rm B}[{\cal R}]]) |0\rangle =
\ee
\be
=\langle 0|  \sum_{m\ge 0}\phi_{\mu-m{\cal R}}h_m(2\pb') \sum_{n\ge 0}\phi_{\nu-n{\cal R}}h_n(2\pb')  |0\rangle
\ee
\be
=2^{-1}Q_{(\mu,\nu)}(\pb^{\rm B})
 \ee
 (here $\mu,\nu$ are numbers). The pfaffian of the Wick theorem yields the projective Schur function labeled by a
 partition $\mu$.

Now we apply the Wick theorem to evaluate
Notice that the number of parts of the partition $R$
Below $\mu=(\mu_1,\dots,\mu_{\kappa^0}),\,\alpha^c =\left(\alpha^c_1,\dots,\alpha^c_{\kappa^c}\right),
\,\beta^c=\left(\beta^c_1,\dots,\beta^c_{\kappa^c}\right)$
denote (strict) partitions.

As we see after re-numbering neutral fermions from complementary groups are quite similar to the charged fermions,
while neutral fermions of the group with $c=0$ up to the re-numbering of the Fourier modes stay
neutral inside VEV (\ref{Q_R}).

Finally applying the Wick theorem to evaluate VEV with three groups of fermions we obtain
\be\label{Q-S}
Q_{R}(\pb^{\rm B}[{\cal R}])=(-1)^\omega 
2^{-\tfrac12{{\bar\ell}(\mu)}}Q_{\mu}(\pb^{\rm B})\prod_{c=1}^{\tfrac 12 ({\cal R}-1)}
s_{(\alpha^c|\beta^c)}(2\pb')
\prod_{i=1}^{\kappa^c}(-1)^{\alpha_i^c+\beta_i^c+c}
\ee
Here $\omega$ depends on the order of embedded parts which belong to one of three groups.
Basically it is not important for our purposes because we will be interested in the rescaling
of the lengths of parts $R_i \to NR_i$ which keeps the order and we get the same $\omega$. However we get the
sign factor due to the fact
that sometimes we get $c \leftrightarrow {\cal R}-c$. We will return to this problem below.

Note that the square of the both sides of (\ref{Q-S}) results in (\ref{S-S}).

\paragraph{Step II}
In what follows we consider the case where $p_k= {\cal R}^{-1}\delta_{k,1}$, 
that is $p^{\rm B}_k[{\cal R}]=\delta_{k,{\cal R}}$.

In this case
\be
Q_{\mu}\{{\cal R}^{-1}\delta_{k,1} \}=\frac{(2{\cal R}^{-1})^{|\mu|}}{\prod_{i=1}^{\kappa^0}\mu_i!}\prod_{i<j}\frac{\mu_i-\mu_j}{\mu_i+\mu_j}
\ee
\be
s_{(\alpha|\beta)}\{2{\cal R}^{-1}\delta_{k,1} \}=(2{\cal R}^{-1})^{|(\alpha|\beta)|}\frac{1}{\prod_{i=1}^{\cal R}\alpha_i!\beta_i!}
\frac{\prod_{i<j}(\alpha_i-\alpha_j)(\beta_i-\beta_j)}{\prod_{i,j}(\alpha_i+\beta_j+1)}
\ee

For such $\tb^{\rm B}[{\cal R}]$ the right hand side of (\ref{Q-S}) is explicit.
The point is the behavior under the rescaling.

\paragraph{Step III. Rescaling of $R$}

We had

The parts of $R$ consist of 3 groups:
 \begin{itemize}
  \item parts ${\cal R}\alpha^0$ that are divisible by ${\cal R}$
  \item parts presented as ${\cal R}\alpha^c+c$ where $c=1,\dots, \frac 12 ({\cal R}-1)$
  \item parts presented as ${\cal R}(\beta^c+1)-c$ where $c=1,\dots, \frac 12 ({\cal R}-1)$
 \end{itemize}

 Now these parts are:
 \begin{itemize}
  \item parts $N {\cal R}\alpha^0$ that are still divisible by ${\cal R}$
  \item parts presented as $N{\cal R}\alpha^c+Nc$ where $Nc=N,\dots, \frac 12 N({\cal R}-1)$
  \item parts presented as $N{\cal R}(\beta^c+1)-Nc$ where $Nc=N,\dots, \frac 12 N({\cal R}-1)$
 \end{itemize}
For each $c$ there exists $q_c$ and $c_N$
\be\label{c-c_N}
Nc= {\cal R} q_c+c_N,\quad c_N < {\cal R}
\ee
Consider the case
\be\label{A}
N<{\cal R}
\ee
\be\label{B}
N\,{\rm and}\,{\cal R}\,{\rm have\,no\,common\,divisor}
\ee
Equation (\ref{c-c_N}) maps each $c$ to certain $c_N$.
Let us show that if $ c $ and $ c'$ are different, then $ c_N (c) $ cannot coincide with
either ${\cal R}-c_N(')$ or $c_N(c')$ .
Indeed,  suppose it is not correct and we have the same $c_N$ for two different $c$ ($c$ and $c'$):
$$
N(c+c')={\cal R}(q_c+q'_c+{\cal R})\,\Rightarrow\,c+c'=\frac{{\cal R}}{N}(q_c+q_c'+1)
$$
which is impossible because of (\ref{A})-(\ref{B}) and because both $c$ and $c'$ are less than $\frac 12 {\cal R}$.
Similarly,
$$
N(c-c')={\cal R}(q_c-q_c')\,\Rightarrow\,c-c'=\frac{{\cal R}}{N}(q_c-q_c')
$$
which is impossible.

Notice that complementary parts defined by $c$ occur to be to complementary again defined by some $c_N$.

\be
(\alpha,\beta)=\left({\cal R}\alpha^c_i +c\,,\,{\cal R}(\beta^c_j+1) -c\right)\,\to\, (N\alpha,N\beta)=
\left(N{\cal R}\alpha^c_i +Nc\,,\,N{\cal R}(\beta^c_j+1) -Nc\right)=
\ee
\be
\left({\cal R}(N\alpha_i^{c}+q_c) +c_N\,,\,{\cal R}(N\beta_j^{c}-q_c) -c_N\right)=
\left({\cal R}{\tilde\alpha}_i^{c_N} +c_N\,,\,{\cal R}({\tilde\beta}_j^{c_N}+1) -c_N\right)=
\ee
where
\be
{\tilde\alpha}^{c_N}_i=N\alpha^c_i +q_c,\quad {\tilde\beta}^{c_N}_j=N\beta^c_j -q_c -1
\ee

Finally we obtain
\be
N\alpha={\cal R}(N\alpha_i +q_i)+c_N,\quad N\beta={\cal R}(N(\beta_i+1) - q_i) -c_N
\ee
 \br
 We can say that there exists a part $R_x$ such that ${\tilde\alpha}^{c_N}=[NR_x {\cal R}^{-1}]$ is the integer part of
 $NR_x{\cal R}^{-1}$ and whose fractial part is $c_N=\{NR_x{\cal R}^{-1}\}$, and there exists another part $R_y$ whose
 $[NR_y{\cal R}^{-1}]$ is $\tilde{\beta}^{c_N}+1$ and whose $\{NR_y{\cal R}^{-1}\}$ is $-c_N$. In this case
 $\langle 0|\gamma(2\pb[{\cal R}])\phi_{NR_x}\phi_{NR_y}|0\rangle \neq 0$.
 \er

Therefore we get
\be
s_{(N\alpha^c|N\beta^c)}\{2\delta_{k,1} \}=
\frac{  (2{\cal R}^{-1})^{|(N\alpha^c|N\beta^c)|} N^{-\kappa^c} }{\prod_{i=1}^{\kappa^c} (N\alpha^c_i+q_i)!(N\beta^c_i -q_i-1) !}
\frac{\prod_{i<j}(\alpha^c_i-\alpha^c_j)(\beta^c_i-\beta^c_j)}{\prod_{i.j=1}^{\kappa^c}(\alpha^c_i+\beta^c_j+1)}
\ee
In a similar was we get
\be
Q_{N\mu}\{\delta_{k,1} \}=\frac{(2{\cal R}^{-1})^{N|\mu|}}{\prod_{i=1}^{\kappa^0} (N\mu_i)!}
\prod_{i<j}\frac{\mu_i-\mu_j}{\mu_i+\mu_j}
\ee

\paragraph{Step IV}

As the result we get
\be
Q_{NR}\{\delta_{k,r}\}=A(R) Q_{R}\{\delta_{k,r}\}
\ee
where
\be
A(R)=
\ee
\be
(-1)^{g}
(2{\cal R}^{-1})^{(N-1)|\mu|+\sum_{c\in {\cal C}} 
\left(|({\tilde\alpha}^{c_N}|{\tilde\beta}^{c_N})| - |(\alpha^c|\beta^c)|\right)}
\left( \prod_{i=1}^{\ell(\mu)}\frac{\mu_i!} {(N\mu_i)!}\right)
\prod_{c\in {\cal C}}N^{-\kappa^c}\prod_{i=1}^{\kappa^c}
\frac{\alpha^c_i!\beta^c_i !}{(N\alpha^c_i+q_c)!(N\beta^c_i-q_c-1)!   }
\ee
\be
=(-1)^g(2{\cal R}^{-1})^w N^{-\tfrac 12 v}\prod_{i=1}^{\ell(R)}\frac{[R_i{\cal R}^{-1}]!}{[NR_i{\cal R}^{-1}]!}
\ee
where $v$ is the number of parts of $R$ that are not divisible by ${\cal R}$, ${\cal C}$ is the set of all numbers $c$
(this set depends on the partition $R$ and the choice of the number ${\cal R}$ and we recall $c\neq 0$), and
where
\be
w=(N-1)|\mu|+\sum_{c\in {\cal C}} \left(|({\tilde\alpha}^{c_N}|{\tilde\beta}^{c_N})| - |(\alpha^c|\beta^c)|\right)
=(N-1)|R|
\ee
$g$ originates from the sign factors in the r h s of (\ref{Q-S}) and is equal to
\be
g=
\sum_{c\in {\cal C}} \sum_{i=1}^{\kappa^c}\left((\alpha^c_i+\beta^c_i+c) -
(\tilde{\alpha}^{c_N(c)}_i+{\tilde\beta}^{c_N(c)}_i+c_N(c))  \right)=(N-1)|R|+\sum_{c\in {\cal C}} \kappa^c(c-c_N(c))
\ee
where $c_N(c)$ is given by (\ref{c-c_N}), and $\kappa^c$ is the number of parts of the partition $R$ whose
remain part after division on ${\cal R}$ is equal to a given $c$. If $c=\{R_i{\cal R}^{-1}\}$
then $c_N(c)=\{NR_i{\cal R}^{-1}\}$.

All these data are contained in a given partition $R$ and the values of $N$ and ${\cal R}$.

\section{A remark on Mironov-Morozov-Natanzon (MMN) cut-and-join relations}

There is a beautifull formula 
\be\label{MM3'}
{\cal W}_{\Delta}(\pb,\partial_{\pb})\cdot s_\mu(\pb)=\tilde{\varphi}_\mu(\Delta)s_\mu(\pb),
\ee
where
\be
\tilde{\varphi}_\mu(\Delta) =\begin{cases}
                              \varphi_\lambda(\Delta),\,|\Delta|=|\lambda|\\
                              \frac{()!}{()!}\varphi_\lambda(\Delta),\,|\Delta|<|\lambda|\\
                              0,\,|\Delta|>|\lambda|
                             \end{cases}
\ee
suggested by Mironov, Morozov and Natanzon \cite{MMN} in the context of  the study of Hurwitz numbers, which 
describes the merging of pairs of branch points in the covering problem. 
Here the dependence of ${\cal W}_{\Delta}$ on $\pb,\partial_{\pb}$ means that ${\cal W}_{\Delta}$ is a
differential operator in the infinitely many power sum variables
$\pb=(p_1,p_2,\dots)$ with coefficients which explicitly depend on $\pb$. It was written
in a normal ordered form that means that in each monomial term all derivates are placed to the right and
the coefficients which depend on $p_1,p_2,\dots$ are placed to the left.
Here $\mu=(\mu_1,\mu_2,\dots)$ and $\Delta=(\Delta_1,\dots,\Delta_\ell )$  are Young diagrams 
($\Delta$ is  the ramification profile of one of branch points,
$s_\mu$ is the Schur function. 
We consider the case where $|\mu|=|\Delta|$), the case $|\Delta|<|\mu|$ is important and is related to the so-called
completed cycles and we do not study it.
The way to construct such operators was suggested in \cite{MMN} and
in more details in the recent preprint \cite{MMZh2021}.

\br
It is more suitable via the characteristic map relation and 
via the Frobenius formula for Hurwitz numbers \cite{ZL}
to rewrite (\ref{MM3'}) in form
\be
{\cal W}_{\Delta}(\pb,\partial_{\pb})\cdot \frac{\pb_{\mu}}{z_\mu}
=\begin{cases}
\sum_{\nu\atop |\nu|=|\mu|} H(\Delta,\mu,\nu)\pb_\nu,\quad |\Delta|=|\mu| \\
\sum_{\nu\atop |\nu|=|\mu|} \tilde{H}(\Delta\cup 1^k,\mu,\nu)\pb_\nu,\quad |\Delta|=|\mu|+k,\,k>0
\end{cases}
\ee
\be
\tilde{H}(\Delta\cup 1^k,\mu,\nu)=\sum_{\lambda\in\Pa\atop |\lambda|=|\mu|+k} 
s_\lambda(\pb_1)\tilde{\varphi}_\lambda(\Delta){\varphi}_\lambda(\mu){\varphi}_\lambda(\nu)
\ee
From
 the geometrical description the case $|\Delta|=|\mu|$ describes the merging of branch points with profiles $\Delta$ and $\mu$.
Here $H(\Delta,\mu,\nu)$ is 3-point Hurwitz number, details see in \cite{MMN},\cite{MMZh2021} and aso in 
\cite{NODubrovin}.
This imply that MMN cut-and-join operators can be written as follows:
\be
{\cal W}_{\Delta}(\pb,\partial_{\pb})=\sum_{k=0}^\infty 
\sum_{\mu,\nu\atop |\mu|=|\nu|=|\Delta|+k}\tilde{ H}(\Delta\cup (1^k),\mu,\nu)\pb_\mu
\left(\partial_\pb \right)_\nu
\ee
where $\left(\partial_\pb \right)_\nu=\frac{\partial}{\partial p_{\nu_1}}\frac{\partial}{\partial p_{\nu_2}}\cdots$
\er

In the work (\cite{MM3}), it was noted that if they are written in the so-called Miwa variables, that is,
in terms of the eigenvalues of the matrix $X$ such that 
$ p_m = \tr \left( X^m \right) $, then
the generalized cut-and-join formula (\ref {MM3'}) is written very compactly and beautifully:
\be\label{MM3}
{\cal W}^{\Delta}\cdot s_\mu(X)=\varphi_\mu(\Delta)s_\mu(X)
\ee
where the Schur function is treated as the homogenious polynomial in the matrix entries of $X$ and where
\be\label{W-Delta}
{\cal W}^{\Delta}={\cal W}^{\Delta}(D)=
\frac{1}{z_\Delta}\vdots\tr\left( D^{\Delta_1} \right)\cdots \tr\left( D^{\Delta_\ell}\right)\vdots,
\ee
 $D$ is the matrix whose entrences are the following diffential operators:
\be\label{D}
D_{a,b}=\sum_{c=1}^N X_{a,c}\frac{\partial}{\partial X_{b,c}}=:\left(X\partial_X\right)_{a,b}
\ee
(As we see $ \tr D $ is the Euler operator which acts on homogenious polynomials of weight $k$ as the multiplication
by $k$. Then $\left( \tr D \right)^q$  acts as the multiplication by $k(k-1)\cdots (k-q+1)$.
This property can be used to study the case of the completed cycles \cite{MMN} where $|\Delta|<|\mu|$.)
The factor $z_\Delta$ is given by $\prod_i m_i!i^{m_i}$ where $m_i$ is the number of parts equal to $i$
in $\Delta$, see \cite{Mac}. The dots $\vdots A\vdots$ denotes the (bosonic) normal ordering of a differential operator
$A$ which means that the derivatives are moved to the right and does not act on the coefficients of the 
differential operator $A$.
It means that each each $X_{a,c}$ in (\ref{D}) can be replaced by the same entry of the any given matrix
$X'$:
\be\label{replacement1}
D_{a,b}=(X\partial_X)_{a,b}\,\to\,(X'\partial_X)_{a,b}:=\sum_{c=1}^N X'_{a,c}\frac{\partial}{\partial X_{b,c}}
\ee
(thus, $\tr D$ is a vector fiels which acts on the entries of $X$) and
\be\label{replacement2}
{\cal W}^{\Delta}(X\partial_X)\,\to\,{\cal W}^{\Delta}(X'\partial_X)
\ee

By the replacement $X\to XC$  we immediately obtain the generalized Mironov-Morozov-Natanzon relation (GMMN)
for the case $|\Delta|=|\mu|$:
\be\label{replacement3}
{\cal W}^{\Delta}(X'\partial_X)\cdot s_\mu(XC)=\varphi_\mu(\Delta)s_\mu(X'C)
\ee
with any $N\times N$ matrix $X'$ and with any $N\times N$ matrix $C$ that is independent of $X$
(the both matrices can be degenerate).
To get (\ref{MM3'}) we chose $X'=X$ and $C=\mathbb{I}_N$.

A repeated application of this formula results in
\be\label{multiGMMN}
\left(
\prod_{i=1}^n
{\cal W}^{\Delta^i}(X'_i\partial_{X_i})\right)\cdot s_\mu(X_1C_1 \cdots X_nC_n)=
\left(\prod_{i=1}^n
\varphi_\mu(\Delta^i)\right) s_\mu(X'_1C_1 \cdots X'_nC_n)
\ee
where $X'_1,\dots, X'_n$ are any $N\times N$ matrices which can depend on $X_1,\dots,X_n$
and where $C_1,\dots,C_n$ are constant matrices (both sets can contain degenerate matrices).
The case where $X'_i=X_i$ describes an eigenproblem.
In \cite{NO2020TMP} it was related to a matrix model based on the sunflower graph with $n$ petels.

As it was pointed out by G.Olshanski relations of type (\ref{MM3}) appear in a different context
 in the works of Perelomov and Popov
\cite {PerelomovPopov1}, \cite {PerelomovPopov2}, \cite {PerelomovPopov3}
and describe the actions of the Casimir operators in the representaion $ \lambda $, see
also \cite{Zhelobenko}, Section 9.

\br
It was shown in \cite{NO2020TMP} that
different graphs result in different generalizations of MMN relation (\ref{MM3}). For instance the polygon
with $n$ vertices drawn in the Riemann sphere yields the relation  
\be\label{multiGMMN-different}
{\cal W}^{\Delta}\left(X'_1\partial_{X_1}\cdots \partial_{X_n}\right)\cdot s_\mu(X_1C_1 \cdots X_nC_n)=
\varphi_\mu(\Delta) s_\mu(X'C_1)\prod_{i=2}^{n} \frac{s_\mu(C_i)}{s_\mu(\pb_1)}
\ee
where $\pb_1=(1,0,0,\dots)$ and
 where $W^\Delta$ is given by (\ref{W-Delta}) however, now, $D$ of (\ref{D}) is replaced by
\be
D_{a,b}= \left(X'\partial_{X_1}\cdots \partial_{X_n}\right)_{a,b}:=\sum_{c^{(1)},\dots,c^{(n-1)}} 
X'_{a,c^{(1)}}\frac{\partial}{\partial (X_1)_{c^{(2)},c^{(1)}}}
\frac{\partial}{\partial (X_2)_{c^{(3)},c^{(2)}}}
\cdots \frac{\partial}{\partial (X_n)_{b,c^{(n-1)}}}
\ee
In case $X'C_1=X_2 C_2 \cdots X_nC_nX_1C_1$ we get the eigenproblem different from (\ref{multiGMMN}) with $X'_i=X_i$.
Formula (\ref{MM3}) is related to the 1-gone case (one vertex, one edge, two faces).
\er

The generalization similar to (\ref{replacement1}),(\ref{replacement2}),(\ref{replacement3}) one can do 
in (\ref{MM3'}) after the replacement of $p_m$ from the set of $\pb=(p_1,p_2,\dots)$
by any given $p_m'$  in each coefficient of the normal ordered expression defining the operator
${\cal W}_{\Delta}(\pb,\partial_{\pb})$:
${\cal W}_{\Delta}(\pb,\partial_{\pb})\,\to\,{\cal W}_{\Delta}(\pb',\partial_{\pb})$ and obtain
\be\label{MM3(p')}
{\cal W}_{\Delta}(\pb',\partial_{\pb})\cdot s_\mu(\pb)=\varphi_\mu(\Delta)s_\mu(\pb'),
\ee
where
\be
{\cal W}_{\Delta}(\pb',\partial_{\pb^{\rm B}})=\sum_{k=0}^\infty\sum_{\mu,\nu\in\Pa\atop |\mu|=|\nu|=k}
\tilde{H}\left(\Delta \cup 1^{k+|\Delta|},\mu,\nu  \right)p'_\mu \left(\partial_{\pb}\right)_\nu
\ee

The way to get relations of type (\ref{multiGMMN-different}) starting from (\ref{MM3(p')}) is yet unclear.

\paragraph{B case.} In \cite{MMN2019} the $B$-type modification of (\ref{MM3'}) was found:

\be\label{MM3'BKP}
{\cal W}^{\rm B}_{\Delta}(\pb^{\rm B},\partial{\pb^{\rm B}})\cdot Q_\mu(\pb^{\rm B})=\varphi^{\rm B}_\mu(\Delta)Q_\mu(\pb^{\rm B}),
\ee
and the explicit construction of operators $ {\cal W}^{\rm B}_{\Delta}(\pb^{\rm B},\partial_{\pb^{\rm B}}) $ was suggested in
\cite{MMN2019}, \cite{MMZh2021}.
In a similar way it results in
\be
{\cal W}^{\rm B}_{\Delta}(\pb^{\rm B},\partial_{\pb^{\rm B}})
=\sum_{\mu,\nu\in\OP\atop |\mu|=|\nu|=|\Delta|}
 2^{\ell(\nu)} H^{\rm B}(\Delta,\mu,\nu)\pb^{\rm B}_\mu \left(\partial_{\pb^{\rm B}}\right)_\nu
\ee
where $H^{\rm B}(\Delta,\mu,\nu)$ are the spin Hurwitz numbers \cite{EOP}, \cite{MMN2019}.

Again, the simple replacement of each $p_i$ from the set $\pb^{\rm B}=(p_1,p_3,p_5,\dots)$ by
$p'_i$ yields in all coefficients of the normal ordered differential operator ${\cal W}^{\rm B}_{\Delta}$
presented in the mentioned works yields
\be\label{MM3'BKP(p')}
{\cal W}^{\rm B}_{\Delta}(\pb',\partial_{\pb^{\rm B}})\cdot Q_\mu(\pb^{\rm B})=\varphi_\mu(\Delta)Q_\mu(\pb'),
\ee
where $\pb'=(p'_1,p'_3,\dots)$ and where
\be
{\cal W}^{\rm B}_{\Delta}(\pb',\partial_{\pb^{\rm B}})=\sum_{}\sum_{}
H^{\rm B}\left(\tilde{\Delta},\mu,\nu  \right)p'_\mu \left(\partial_{\pb^{\rm B}}\right)_\nu
\ee

To my knowledge the matrix realization of this spin cut-and-join equation is yet unkown.

\bigskip
\bigskip
\noindent
\small{ {\it Acknowledgements.}
The author is grateful to J. Harnad for the collaboration in the work \cite{paper1},\cite{paper2},\cite{paper3}
which results in the present research, J.Harnad and J. van de Leur for various discussions concerning BKP,  
to  A.Morozov, A. Mironov and A. Alexandrov, for
attracting attention to \cite{Alex2},\cite{MMq} and special thanks to A. Mironov and J.Harnad for usefull
remarks concerning the present paper and A.Mironov for additional technical help.
 The work was supported by the Russian Science
Foundation (Grant No.20-12-00195).

\bigskip



\appendix

\section{The key lemma \cite{paper1}}
\label{polarizations}

\begin{definition}
A {\em polarization} of $(\alpha|\beta)=(\alpha_1,\dots,\alpha_r|\beta_1,\dots,\beta_r)$,  is a pair  $(\nu^+, \nu^-)$
of strict partitions with cardinalities  (or {\em lengths})
\be
m(\nu^+):= \#(\nu^+), \quad  m(\nu^-)):= \#(\nu^-)
\label{m_pm_def}
\ee
(including  possibly a zero part $\nu^+_{m(\nu^+)}=0$ or $\nu^-_{m(\nu^-)}=0$), satisfying
\be
\nu^+ \cap \nu^- = \alpha \cap I(\beta), \quad \nu^+ \cup \nu^- = \alpha \cup I(\beta),
\label{mu_pm_cap_cup}
\ee
where
\be
I(\beta) := (I_1(\beta), \dots I_r(\beta))
\label{I_beta}
\ee
is the strict partition \cite{Mac} with parts
\be
I_j(\beta) = \beta_j +1, \quad j=1, \dots r.
\label{I_j_beta}
\ee
The set of all polarizations of $(\alpha|\beta)$  is denoted $\PP(\alpha,\beta)$.
\end{definition}
We denote the strict partition obtained by intersecting  $\alpha$ with $I(\beta)$ as
\be
S := \alpha \cap I(\beta),
\label{S_def}
\ee
and its cardinality as
\be
s := \#(S).
\ee
Since both $\alpha$ and $I(\beta)$ have cardinality $r$, it follows that
\be
m(\nu^+) + m(\nu^-) = 2r,
\label{m_pm_sum}
\ee
so $m(\nu^\pm)$ must have the same parity.   It is easily verified \cite{paper1} that the cardinality of
$\PP(\alpha, \beta)$ is $2^{2r -2s}$.
The following was proved in \cite{paper1}.
\begin{lemma}[\bf Binary sequence associated to a polarization]
For every polarization $\nu:=(\nu^+, \nu^-)$ of $\lambda=(\alpha| \beta)$ there is a unique binary sequence
of length $2r$
\be
\epsilon(\nu) =(\epsilon_1(\nu), \dots, \epsilon_{2r}(\nu)),
\label{epsilon_mu_binary}
\ee
with
\be
\epsilon_j(\nu) =\pm, \quad j=1, \dots 2r,
\label{epsilon_j_binary}
\ee
such that
\begin{enumerate}
\item The sequence of pairs
\be
((\alpha_1, \epsilon_1(\nu)), \cdots  (\alpha_r, \epsilon_r(\nu)),
(\beta_1+1, \epsilon_{r+1}(\nu)), \dots , ( \beta_r+1, \epsilon_{2r}(\nu))
\label{alpha_beta_sequence}
\ee
is a permutation of the sequence
\be
((\nu^+_1,+), \cdots ( \nu^+_{m^+(\nu)},+), (\nu^-_1, -),  \dots  , (\nu^-_{m^-(\nu)}, -)
\label{mu_sequence}
\ee
\item
\be
 \epsilon_j(\nu)= + \ \text{ if } \alpha_j \in S, \ \text{ and }
 \epsilon_{r+j}(\nu)= - \ \text{ if } \beta_j +1\in S, \quad j=1, \dots, r.
 \label{canon_epsilon_constraint}
 \ee
\end{enumerate}
\end{lemma}

\begin{definition}
 The {\em sign} of the polarization $(\nu^+, \nu^-)$, denoted $\sgn(\nu)$, is defined
as the sign of the permutation  that takes the sequence (\ref{alpha_beta_sequence})  into the sequence
(\ref{mu_sequence}).
\end{definition}
Denote by
\be
\pi(\nu^\pm) :=\#(\alpha \cup \nu^{\pm})
\label{pi_pm_mu_def}
\ee
the cardinality of the intersection of $\alpha$ with $\nu^{\pm}$. It follows that
\be
\pi(\nu^+) + \pi(\nu^-) = r+s.
\label{pi_pm_sum}
\ee
 We then have
 \begin{lemma}
\be
|\lambda \rangle =  {(-1)^{\tfrac{1}{2}r(r+1) + s}\over 2^{r-s}} \sum_{\nu\in \PP(\alpha, \beta)}
\sgn(\nu)(-1)^{\pi(\nu^-)} i^{ m(\nu^-)}\prod_{j=1}^{m(\nu^+)} \phi^+_{\nu^+_j} \prod_{k=1}^{m(\nu^-)} \phi^-_{\nu^-_k}  |0 \rangle.
\label{polariz_sum}
\ee
\end{lemma}
\begin{proof}
 For the proof we reorder the product over the factors $\psi_{\alpha_j} \psi^\dag_{-\beta_j-1}$   so
 the  $\psi_{\alpha_j}$ terms precede the $\psi^\dag_{-\beta_j-1}$ ones, giving an overall sign
 factor $(-1)^{\tfrac{1}{2}r(r-1)}$.  Then substitute
 \be
\psi_{\alpha_j} = \tfrac{1} {\sqrt{2}}( \phi^+_{\alpha_j} - i \phi^-_{\alpha_j}), \quad
 \psi^\dag_{-\beta_j-1}= \tfrac{(-1)^j}{ \sqrt{2}} (\phi^+_{\beta_j+1} + i \phi^-_{\beta_j+1}), \quad j\in \Zb,
\ee
in each factor, and expand  the product as a sum over monomial terms of the form
\be
 \prod_{j=1}^{m(\nu^+)} \phi^+_{\nu^+_j} \prod_{k=1}^{m(\nu^-)} \phi^-_{\nu^-_k}  |0 \rangle.
 \ee
 Taking into account the sign factor $\sgn(\nu)$ corresponding to the order of the neutral fermion
factors, as well as the powers of $-1$ and $i$, and noting that there are $2^s$  resulting identical terms,
then gives (\ref{polariz_sum}).
\end{proof}

\begin{definition}{\bf Supplemented partitions.}
If $\nu$ is a strict partition of cardinality $m(\nu)$ (with $0$ allowed as a part),
we define the associated {\em supplemented partition} $\hat{\nu}$ to be
\be
\hat{\nu} := \begin{cases}  \nu \ \text{ if } m(\nu) \ \text{ is even}, \cr
                  (\nu,0)  \  \text{ if } m(\nu) \  \text{ is odd}.
                  \end{cases}
 \label{hat_nu'}
\ee
We denote by $m(\hat{\nu})$  the cardinality of $\hat{\nu}$.
\end{definition}

\section{Sergeev group \cite{EOP},\cite{Serg}}

As it was introduced in \cite{EOP} the Sergeev group ${\bf C}(d)$ is the semidirect product  of the permutation group
$S_d$ and the Clifford group ${\rm Cliff}_d$ generated by the
involutions $\xi_i$, $i=1,\dots,d$ and the central involution $\epsilon$ which are the subjects of
$$
\xi_i\xi_j=\epsilon \xi_j\xi_i
$$
The group $S_d$ acts on ${\rm Cliff}_d$ by the permutation of $\xi_i$'s. The irreducible representations
of ${\bf C}(d)$ are labeled with strict partitions. The normalized characters 
$$
\textbf{f}_\alpha(\Delta) :=
2^{-\ell(\Delta)}\langle 0|J^{\rm B}_{-\Delta}\Phi_{\alpha} |0\rangle \frac{1}{z_\Delta}\frac{1}{Q_\alpha\{\delta_{k,1}\}}
=\frac{2^{\ell(\alpha)+\ell(\Delta)}}{z_\Delta Q_\alpha\{\delta_{k,1}\} }
\chi^{\rm B}_\alpha(\Delta)
$$
(here $z_\Delta=\prod_i m_i!i^{m_i}$ and $\Delta=(1^{m_12^{m_2}}\cdots)\in\OP$, see \cite{Mac})
which is related to $Q$ functions by
\be
Q_\alpha\{p_k\}=Q_\alpha\{\delta_{k,1}\}\sum_{\Delta\in\OP} \textbf{f}_\alpha(\Delta) \pb_\Delta
\ee
(this relation is dual to (\ref{pbB-Q}))
enter the formula for spin Hurwitz numbers
$$
H^\pm\left(\Delta^1,\dots,\Delta^k\right) =
\sum_{\alpha\in\DP\atop |\alpha|=d,\ell(\alpha)\,{\rm even}} \left( Q_\alpha\{\delta_{k,1}\}\right)^2
\textbf{f}_\alpha(\Delta^1)\cdots \textbf{f}_\alpha(\Delta^k)
$$
$$
-
\sum_{\alpha\in\DP\atop |\alpha|=d,\ell(\alpha)\,{\rm odd}} \left( Q_\alpha\{\delta_{k,1}\}\right)^2
\textbf{f}_\alpha(\Delta^1)\cdots \textbf{f}_\alpha(\Delta^k)
$$
see \cite{EOP},\cite{Lee2018},\cite{MMN2019},\cite{MMNO}.
Then one can relate the evaluation of the spin Hurwitz numbers to the BKP tau functions in
the similar way as the evaluation of usual Hurwits numbers is related to KP and Toda lattice theory,
see \cite{Okoun2000},\cite{Okounkov-Pand-2006},\cite{GJ},\cite{MMN2011}.

\section{Scalar products and matrix integrals}

In the topics related to symmentric function the notion of the scalar product is important, see \cite{Mac}.
Here we study the realization of the scalar products as matrix integrals. We consider cases (A) and (B):

(A) Take

\be\label{forSchur}
<\pb_{\Delta^1},\pb_{\Delta^2} > = z_{\Delta^1} \delta_{\Delta^1,\Delta^2},\quad \Delta^1,\Delta^2\in \Pa
\ee
If (\ref{forSchur}) is chosen, then (see \cite{Mac}, Section )
\be\label{Schur-orth}
<s_\lambda(\pb),s_\mu(\pb)> = \delta_{\lambda,\mu},\quad \lambda,\mu\in \Pa
\ee

(B) Another choice
\be\label{forSchurB}
<\pb_\lambda,\pb_\mu >_{\rm B} = 2^{-\ell(\lambda)}z_\lambda \delta_{\lambda,\mu},\quad \lambda,\mu\in \OP
\ee
results in
\be\label{Schur-orthB}
<Q_\lambda(\pb),Q_\mu(\pb)>_{\rm B} = 2^{\ell(\lambda)}\delta_{\lambda,\mu},\quad \lambda,\mu\in \DP
\ee
see Section  in \cite{Mac}.

\paragraph{Case (A)}:

Take
\be
p_m =p_m(\xb)=\tr X^m = \sum_{i=1}^N  x_i^N
\ee
Suppose $U_1,U_2\in\mathbb{U}(N)$. 
The Haar measure on $\mathbb{U}(N)$ in the explicit form is
\be\label{Haar-U}
d_*U=(2\pi)^N\prod_{i<j}|e^{\theta_i}-e^{\theta_j}|^2 \prod_{i=1}^N d\theta_i
\ee
where $e^\theta_i$ are eigenvalues of $U$.

Then we have
\begin{proposition} The scalar product (\ref{Schur-orth}) can be realized as follows: 
\be\label{ss=delta}
\int_{\mathbb{U}(N)\times\mathbb{U}(N) } s_\lambda(U_1)s_\mu(U_2) \det
\left(\mathbb{I}_N-U^{-1}_1 U^{-1}_2\right)^{-N} d_*U_1d_*U_2
=\delta_{\lambda,\mu}
\ee
\be\label{pp=z-delta}
\int_{\mathbb{U}(N)\times\mathbb{U}(N)} \pb_{\Delta^1}(U_1)\pb_{\Delta^2}(U_2) \det
\left(\mathbb{I}_N-U^{-1}_1 U^{-1}_2\right)^{-N} d_*U_1d_*U_2=z_{\Delta^1} \delta_{\Delta^1,\Delta^2}
\ee 
\end{proposition}
The proof of (\ref{ss=delta}) follows from
\be\label{1-AB}
 \det\left(\mathbb{I}_N-AB\right)^{-N} =
\sum_{\lambda\in\Pa} s_\lambda(AB)s_\lambda(\mathbb{I}_N)
\ee
and from
\be\label{s(AUBU^{-1})}
\int_{U\in\mathbb{U}(N)} s_\lambda(AUBU^{-1})) d_*U=\frac{s_\lambda(A)s_\lambda(B)}{s_\lambda(\mathbb(\mathbb{I}_N))}
\ee
where we put $A=U^{-1}_1,\,B= U^{-1}_2$. The proof of (\ref{pp=z-delta}) follows from the linear relation 
(the characteristic map relation) between $\{s_\lambda \}$ and $\{ \pb_\Delta \}$ polynomials.

Next, suppose $Z_1,Z_2 \in \mathbb{GL}_N(\mathbb{C})$. 
Consider the measure
\be
d\mu(Z)=e^{-\tr \left(ZZ^\dag \right) }\prod_{i,j=1}^N d\Re Z_{ij} d\Im Z_{ij}
\ee
Let us note that if we use the Schur factorization of the complex matrix $Z=U\left(D+B\right)U^\dag$ where
$U\in\mathbb{U}(N)$, $D$ is diagonal and $B$ is strictly upper trianguler, then the measure can be written
as
\be\label{measure-eigenvalues}
d\mu(Z)=d_*U  e^{-\tr BB^\dag} \prod_{i=1}^N dD_{ii} e^{-(D_{ii})^2}\prod_{i<j} B_{ij}
\ee

We have

\begin{proposition}
\be\label{ss=intGL}
\int_{\mathbb{GL}_N\times\mathbb{GL}_N } s_\lambda(Z_1)s_\nu(Z_2) {\cal F}(Z^\dag_1 Z^\dag_2)
 d\mu(Z_1)d\mu(Z_2)
=\delta_{\lambda,\nu}
\ee
\be\label{pp=intGL}
 \int_{\mathbb{GL}_N\times \mathbb{GL}_N} \pb_{\Delta^1}(Z_1)\pb_{\Delta^2}(Z_2) {\cal F}(Z^\dag_1 Z^\dag_2)
 d\mu(Z_1)d\mu(Z_2) =z_{\Delta^1} \delta_{\Delta^1,\Delta^2}
\ee
where
\be\label{F}
{\cal F}(Z^\dag_1 Z^\dag_2)=
\sum_{\mu\in\Pa} \frac{s_\mu(Z^\dag_1 Z^\dag_2)s_\mu(\mathbb{I}_N)}{\left((N)_\mu\right)^2},
\quad (N)_\mu:=\prod_{(i,j)\in\mu}\left(N+j-i\right)
\ee
\end{proposition}
Suppose $Z_i=U_i\left(D_i+B_i\right)U^\dag_i$, see (\ref{measure-eigenvalues}). Let us integrate the integrand
of (\ref{ss=intGL}) over $U=U_1U^\dag_2$ (namely integrate the factor ${\cal F}$)
with the help of (\ref{s(AUBU^{-1})}) where we take $A=Z^{\dag}_1,\,B=Z^{\dag}_2$. 
Along this action the factor $s_\lambda(\mathbb{I}_N)$ disappears and the whole integral turns to be the sum over 
partitions $\lambda$ from the relation (\ref{F}) of a  product of two integrals (integrals over $Z_1$ and over $Z_2$).
Then we apply the orthogonality relation:
\be
\int_{\mathbb{GL}_N} s_\lambda(Z_i)s_\nu(Z^\dag_i)
 d\mu(Z_i) = (N)_\lambda \delta_{\lambda,\nu},\quad i=1,2
\ee
and after the cancellation of the factor $(N)_\lambda$ we obtain (\ref{ss=intGL}). 

The characteristic map relation gives (\ref{pp=intGL}).

Let us note that ${\cal F}$ is known hypergeometric function of matrix argument and is an example of the KP
tau function \cite{OS2000}.

\paragraph{Case B.}
Here we will show that an analogue of Propostition based on the similar idea to present the B-type
scalar product as $2N$-fold matrix integral fails.
 
Take
\be
p^{\rm B}_m =p^{\rm B}_m(\xb)=2\tr X^m = 2\sum_{i=1}^N  x_i^m,\quad m\,{\rm odd},
\ee
where $\xb=(x_1,\dots.x_N)$ are eigenvalues of $X$.
We introduce the notation $Q_\alpha(X):=Q_\alpha(\pb^{\rm B}(\xb))=Q_\alpha(\xb)$.

Suppose $U_1,U_2\in\mathbb{U}(N)$. 

Using the anti-symmetry of determinants and change of variables and also (\ref{Haar-U}) 
(in the same way as (A.34) was derived in \cite{HO-2005})) we obtain
$$
\int_{\mathbb{U}(N)\times\mathbb{U}(N) } Q_\alpha(U_1)Q_\beta(U_2) \det
\left(\mathbb{I}_N-U^{-2}_1 U^{-2}_2\right)^{-N} d_*U_1d_*U_2=
$$
$$
(2\pi^2)^{\frac12 N(n-1)}\prod_{k=1}^{N-1}\frac{1}{k!}
\oint\cdots\oint Q_\alpha(\xb)Q_\beta(\yb) \Delta^{\rm B}(\xb)\Delta^{\rm B}(\yb) 
\prod_{i=1}^N \frac {x^{-1}_i  y^{-1}_i d x_i d y_i}{1-x^{-2}_iy^{-2}_i}
$$
where
\be
\Delta^{\rm B}(\xb):=\prod_{i<j}\frac{x_i-x_j}{x_i+x_j}
\ee
However, this scalar product is obviously degenerate because for $\alpha=\beta=(1)$ it is equal to $0$.

To realize the scalar product as an integral over matrices we could conjecture that we need
the integral over Sergeev supergroup $\textsc{q}(N)$.

\end{document}